\pdfoutput=1
\documentclass[amsmath,amssymb, aps, twocolumn, superscriptaddress, longbibliography]{revtex4-2}

\usepackage[verbose=true]{microtype}
\usepackage{amsfonts}
\usepackage{amsthm}
\usepackage{mathtools}
\usepackage[colorlinks,linkcolor=black,citecolor=blue,linktocpage,hypertexnames=true,pdfpagelabels]{hyperref}
\usepackage[capitalise]{cleveref}
\usepackage{array}
\usepackage{enumerate}
\usepackage{enumitem}
\usepackage{tikz}
\usepackage{braket}
\usepackage{caption}
\usepackage{xcolor}
\usepackage{adjustbox}
\usepackage{transparent}

\usepackage{scalerel}

\usepackage[super]{nth}

\DeclareMathOperator{\tr}{tr}

\usepackage[small, langfont=caps]{complexity}

\newlang{\HALT}{Halt}
\newlang{\BHALT}{BHalt}
\newlang{\NHALT}{NHalt}
\newlang{\BNHALT}{BNHalt}
\newlang{\NHALTAll}{NHaltAll}
\newlang{\BNHALTAll}{BNHaltAll}

\newlang{\LangPCP}{Pcp}
\newlang{\BPCP}{BPcp}

\newlang{\MM}{Mm}
\newlang{\BMM}{BMm}

\newlang{\ZULC}{Zulc}
\newlang{\BZULC}{BZulc}

\newlang{\MPO}{Mpo}
\newlang{\BMPO}{BMpo}

\newlang{\POLY}{Poly}
\newlang{\BPOLY}{BPoly}

\newlang{\STAB}{Tsp}
\newlang{\BSTAB}{BTsp}

\newlang{\TILE}{Tile}
\newlang{\BTILE}{BTile}

\newlang{\REACH}{Reach}
\newlang{\BREACH}{BReach}

\newlang{\MAXCUT}{MaxCut}

\newlang{\BGSE}{BGse}
\newlang{\GSE}{Gse}

\usepackage{dsfont}
\usepackage{eufrak}

\usepackage[utf8x]{inputenc}
\usepackage[T1]{fontenc}
\usepackage{lmodern}

\newcommand{\gstar}{\textcolor{gray}{\star}}

\newtheorem{theorem}{Theorem}

\newtheorem{definition}[theorem]{Definition}

\usetikzlibrary{shadows.blur}
\usetikzlibrary{shapes.symbols}

\renewenvironment{quote}{%
  \list{}{%
    \leftmargin0.3cm
    \rightmargin\leftmargin
  }
  \item\relax
}
{\endlist}

\usepackage{todonotes}

\definecolor{color1}{HTML}{7FB77E}
\definecolor{color2}{HTML}{B1D7B4}
\definecolor{color3}{HTML}{F7F6DC}
\definecolor{color4}{HTML}{FFC090}

\usepackage[normalem]{ulem}

\begin{document}
\title{Many bounded versions of undecidable problems are NP-hard}

	\author{Andreas Klingler}
	\thanks{These authors contributed equally to this work. \\Email: \href{mailto:Andreas.Klingler@uibk.ac.at}{Andreas.Klingler@uibk.ac.at}, \\{\phantom{Email: }}\href{mailto:Mirte.van-der-Eyden@uibk.ac.at}{Mirte.van-der-Eyden@uibk.ac.at}}
	\affiliation{Institute for Theoretical Physics, Technikerstr.\ 21a,  A-6020 Innsbruck, Austria}

	\author{Mirte van der Eyden}
	\thanks{These authors contributed equally to this work. \\Email: \href{mailto:Andreas.Klingler@uibk.ac.at}{Andreas.Klingler@uibk.ac.at}, \\{\phantom{Email: }}\href{mailto:Mirte.van-der-Eyden@uibk.ac.at}{Mirte.van-der-Eyden@uibk.ac.at}}
	\affiliation{Institute for Theoretical Physics, Technikerstr.\ 21a,  A-6020 Innsbruck, Austria}

	\author{Sebastian Stengele}
	\affiliation{Institute for Theoretical Physics, Technikerstr.\ 21a,  A-6020 Innsbruck, Austria}

	\author{Tobias Reinhart}
	\affiliation{Institute for Theoretical Physics, Technikerstr.\ 21a,  A-6020 Innsbruck, Austria}
	
	\author{Gemma De las Cuevas}
	\affiliation{Institute for Theoretical Physics, Technikerstr.\ 21a,  A-6020 Innsbruck, Austria}	
	
	\date{\today}

\begin{abstract}
Several physically inspired problems have been proven undecidable; examples are the spectral gap problem and the membership problem for quantum correlations. Most of these results rely on reductions from a handful of undecidable problems, such as the halting problem, the tiling problem, the Post correspondence problem or the matrix mortality problem. All these problems have a common property: they have an $\NP$-hard bounded version. This work establishes a relation between undecidable unbounded problems and their bounded  $\NP$-hard versions. Specifically, we show that $\NP$-hardness of a bounded version follows easily from the reduction of the unbounded problems. This leads to new and simpler proofs of the $\NP$-hardness of bounded version of the Post correspondence problem, the matrix mortality problem, the positivity of matrix product operators, the reachability problem, the tiling problem, and the ground state energy problem. 
This work sheds light on the intractability of problems in theoretical physics and on the computational consequences of bounding a parameter. 
\end{abstract}

\maketitle

\section{Introduction}
\label{sec:introduction}

Many problems in quantum information and quantum many-body physics are undecidable. 
This includes the spectral gap of  physical systems \cite{Cu15,Ba18b}, membership problems for quantum correlations \cite{Sl19, Sl19b, Ji21, Fu21, Mo21}, properties of tensor networks \cite{De15, Kl14, Sc20}, measurement occurrence and reachability problems \cite{Ei12, Wo11}, and many more \cite{De21b, Ey21, Bl03, Sc21b, Fr16}. In addition, other problems are believed to be undecidable, such as detecting quantum capacity \cite{Cu14}, distillability of entanglement \cite{Wo11}, or tensor-stable positivity \cite{Ey21}.

All these problems have a common theme: They ask for a property that includes an unbounded parameter. For example, in a quantum correlation scenario, the dimension of the shared quantum state between the two parties may be unbounded. Also properties of many-body systems, like the spectral gap, are statements involving arbitrarily large system sizes.

On the other hand, many problems in science, engineering, and mathematics are $\NP$-hard \cite{Wi19b}.  
Some examples relevant for physics are finding the ground state energy of an Ising model \cite{Ba82}, 
the training of variational quantum algorithms \cite{Bi21}, or the quantum separability problem \cite{Gu03,Gh09}, and many more \footnote{There are thousands of NP-complete problems. More than three hundred of them are presented in Ref.\ \cite{Ga79}.}.
These problems typically concern properties where all size parameters are bounded or even fixed. For example, the ground state energy problem concerns the minimal energy of Hamiltonians with fixed system size.

This highlights an analogy between certain classes of problems: 
an \emph{unbounded problem} tests a property for an unbounded number of occurrences (which can be generated recursively), 
whereas the corresponding \emph{bounded version} tests the same property for a bounded number of situations. 
This includes, for example, testing a certain property of a translational  invariant spin system for all system sizes, or up to a certain size.  
One common observation in this context is that bounded versions of undecidable problems are $\NP$-hard. 
This observation was already made in \cite{Kl14, Bl97, Sc20} for various examples, as well as in \cite[Chapter 3]{Wi19b}.

Despite this analogy, the techniques used to prove $\NP$-hardness and undecidability often differ. While proofs of undecidability mainly rely on reductions from the halting problem, the Post correspondence problem or the Wang tiling problem, $\NP$-hardness proofs mainly rely on reductions from the satisfiability problem $\SAT$, or from $\NP$-complete graph problems like the $3$-coloring problem or $\MAXCUT$.

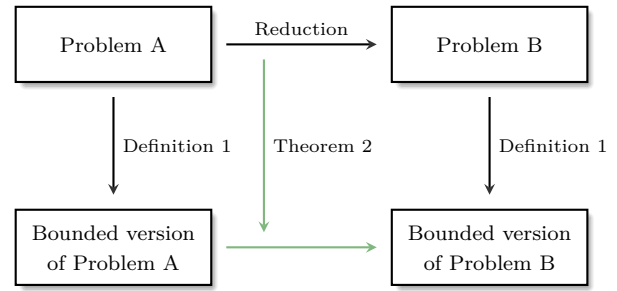
\begin{figure}[t]\centering
\begin{tikzpicture}

\draw[shade, blur shadow={shadow blur steps=8,shadow blur extra rounding=1.3pt, shadow scale=0.98}] (-1.3,2) rectangle (1.3,3);
\draw[thick, fill=white] (-1.3,2) rectangle (1.3,3);

\draw[shade, blur shadow={shadow blur steps=8,shadow blur extra rounding=1.3pt, shadow scale=0.98}] (-1.3,-0.7) rectangle (1.3,0.3);
\draw[thick, fill=white] (-1.3,-0.7) rectangle (1.3,0.3);

\draw[shade, blur shadow={shadow blur steps=8,shadow blur extra rounding=1.3pt, shadow scale=0.98}] (3.7,2) rectangle (6.3,3);
\draw[thick, fill=white] (3.7,2) rectangle (6.3,3);

\draw[shade, blur shadow={shadow blur steps=8,shadow blur extra rounding=1.3pt, shadow scale=0.98}] (3.7,-0.7) rectangle (6.3,0.3);
\draw[thick, fill=white] (3.7,-0.7) rectangle (6.3,0.3);

\draw node at (0,2.5) {\footnotesize Problem A};

\draw node at (5,2.5) {\footnotesize Problem B};

\draw node at (0,0) {\footnotesize Bounded version};
\draw node at (0,-0.4) {\footnotesize of Problem A};

\draw node at (5,0) {\footnotesize Bounded version};
\draw node at (5,-0.4) {\footnotesize of Problem B};

\draw[thick,  -stealth] (1.5,2.5) -- (3.5,2.5) node[midway, above] {\scriptsize Reduction};
\draw[thick, color1, -stealth] (1.5,-0.2) -- (3.5,-0.2) node[midway, below] {\scriptsize };

\draw[thick,  -stealth] (5,1.8) -- (5,0.5) node[midway, right] {\scriptsize \cref{def:boundedVersion}};
\draw[thick,  -stealth] (0,1.8) -- (0,0.5) node[midway, right] {\scriptsize \cref{def:boundedVersion}};

\draw[thick, color1, -stealth] (2,2.3) -- (2,0) node[black, midway, right] {\scriptsize \cref{thm:boundedHardness}};
%\node at (2.6,1) {\scriptsize Thm. 1};

\end{tikzpicture}
\caption{If Problem B is at least as hard as Problem A (i.e.\ there is a reduction from A to B), 
is the bounded version of Problem B at least as hard as the bounded version of Problem A?   
\Cref{thm:boundedHardness} gives a sufficient condition  
when this is the case by reusing the reduction between their unbounded versions. }
\label{fig:intro}
\end{figure}

In this work, we establish a relation between undecidable problems and certain $\NP$-hard problems. 
Specifically, we define a bounded version of a problem 
and a method to leverage the reduction from unbounded problems to their corresponding bounded problems (see \Cref{fig:intro}).
Then we present two versions of the halting problem whose bounded versions are $\NP$-hard, 
and use these, together with our method, 
to provide simple and unified proofs of the $\NP$-hardness of the bounded version of
the Post correspondence problem, 
the matrix mortality problem, 
the positivity of matrix product operators, 
the reachability problem, 
the tiling problem, and  
the ground state energy problem.  

This work sheds light on the various intractability levels of problems used in theoretical physics by     
highlighting the computational consequences of bounding a parameter.
More generally, this work is part of a tradition of studying problems from a  computational perspective, which   has proven extremely successful in mathematics and beyond
\cite{Wi19b}. For example, the hardness results of the ground state energy problem rule out a tractable solution of the ground state for a given Hamiltonian, both for unbounded system sizes as well as a fixed system size.

This paper is structured as follows. 
In \Cref{sec:framework}, we present a definition of bounding 
and a method to leverage the reduction from unbounded problems to their corresponding bounded versions. 
In \Cref{sec:root}, we present the two halting problems which are the root undecidable problems. 
In \Cref{sec:examples}, we apply our framework to many examples. 
In \Cref{sec:questions}, we conclude and discuss future directions. 
The appendix contains basic background on computational complexity (\Cref{app:comput}), 
the hardness proofs of the root undecidable problems (\Cref{app:HaltingProofs}) 
and more details on the discussed examples (\Cref{app:Examples}). 

\section{Bounding}
\label{sec:framework}

In this section, we  present a definition of a bounded version of a language (\cref{ssec:definition}), and  a method to leverage the reduction from unbounded problems to their corresponding bounded versions (\cref{ssec:leveraging}).  
For a short introduction to computational complexity, we refer the reader to \cref{app:comput}. To the best of our knowledge, no prior work introduces or studies bounded versions of problems from a general systematic perspective.

\subsection{Definition of bounding}
\label{ssec:definition}

Let $\Sigma$ be a finite alphabet and $\Sigma^{*}$ the set of all words generated from $\Sigma$. A language $L \subseteq \Sigma^{*}$ encodes all the yes-instances of a given problem, i.e.\ $x \in L$ if $x$ is a yes-instance and $x \notin L$ if $x$ is a no-instance.

We now define a bounded version $L_B$ of $L$. For this purpose, we add a second parameter $n \in \mathbb{N}$ to every yes-instance in $L$. This parameter acts as an acceptance threshold for every yes-instance $x \in L$ and is encoded in unary, i.e.\ for $1 \in \Sigma$, every element of $L_B$ is of the form $\langle x, 1^n \rangle$, where $1^n$ represents the $n$-fold concatenation of $1$.

\begin{definition}
\label{def:boundedVersion}
Let $L \subseteq \Sigma^{*}$ be a language. A language
$$L_B \subseteq \big\{\langle x, 1^n \rangle \mid x \in \Sigma^{*}, n \in \mathbb{N}\big\}$$
is called a \emph{bounded version of} $L$ if
\begin{enumerate}[label=(\roman*)]
	\item\label{cond:1} $x \in L \quad \Longleftrightarrow \quad \exists n \in \mathbb{N}: \langle x, 1^n \rangle \in L_B$.
	\item\label{cond:2} $\langle x, 1^n \rangle \in L_B \quad \Longrightarrow \quad \langle x, 1^{n+1}\rangle \in L_B$.
\end{enumerate}
\end{definition}

We shall often refer to $L$ as the \emph{unbounded} language of $L_B$.  

First, note that the definition of bounded versions relies only on the existence of a parameter $n$ in the problem that acts accordingly. While most problems we consider in this paper are $\RE$-complete, \Cref{def:boundedVersion} applies to languages of arbitrary complexity. Moreover, note that the bounding parameter can also be encoded differently. For example, if the parameter is encoded in binary, most of the bounded version would be $\NEXP$-hard instead of $\NP$-hard. Finally, we remark that the process of bounding a language can be reversed. Given a language $L_B$ with instances of the form $\langle x, 1^n\rangle$ satisfying only Condition \ref{cond:2}, there is a unique language $L$, defined via \ref{cond:1}, which is the unbounded language of $L_B$.

Many problems mentioned in the introduction contain a parameter that gives rise to a bounded version according to \Cref{def:boundedVersion}. This parameter can be the system size for tensor network and spectral gap problems, or the dimension of the entangled state for quantum correlation scenarios; we will present many such examples in \Cref{sec:examples}. 

As an example, let us consider the halting problem $\HALT$ with its known bounded version $\BHALT$. The former takes instances $\langle T,x_0 \rangle$ with a description $T$ of a Turing machine and an input $x_0$. An instance $\langle T, x_0\rangle$ is accepted if and only if the Turing machine $T$ halts on $x_0$. The bounded halting problem takes instances $\langle T,x_0,1^n \rangle$, which are accepted if and only if the Turing machine halts on $x_0$ within $n$ computational steps. $\BHALT$ is indeed a bounded version according to \Cref{def:boundedVersion} since halting of a Turing machine is equivalent to the existence of a finite halting time, and halting within $n$ steps implies halting within $n+1$ steps. 

We remark that in \cref{def:boundedVersion} there is some freedom in the choice of the bounding parameter. For example, for every non-decreasing, unbounded function $f: \mathbb{N} \to \mathbb{N}$, the language
$$\BHALT_f \coloneqq \big\{\langle T,x_0,1^n \rangle \mid T \text{ halts on } x_0 \text{ in } f(n) \text{ steps} \big\}$$
is also a bounded version of $\HALT$. In this paper, we will focus on the simplest versions setting $f = \textrm{id}$ in all examples.

\subsection{Leveraging reductions to the bounded case}
\label{ssec:leveraging}

Given the hardness of the unbounded languages, what can we say about the bounded ones?  
We will now give a condition to leverage a reduction of unbounded problems to a reduction between the  corresponding bounded problems.
This results in a method to prove hardness results of many bounded versions of undecidable problems, as we will see in \Cref{sec:examples}.

Let $L_B$ be a bounded version of $L \subseteq \Sigma^{*}$. 
For $x \in \Sigma^{*}$, we define the threshold parameter
$$
n_{\min, L}[x] \coloneqq \inf \{n \in \mathbb{N}: \langle x, 1^n \rangle \in L_B \}
$$
where we set $\inf \emptyset = \infty$. In other words, $n_{\min}[x]$ denotes the minimum value of $n$ leading to an accepting instance of $L_B$.  
Note that $n_{\min}[x] < \infty$ for every $x \in L$ due to \ref{cond:1} of \Cref{def:boundedVersion} and $n_{\min}[x] = \infty$ if $x \notin L$. Moreover, $\langle x,1^n \rangle \in L_{B}$ if and only if $n \geq  n_{\min}[x]$ due to \ref{cond:2} of \Cref{def:boundedVersion}.

\begin{theorem}
\label{thm:boundedHardness}
Let $L_1,L_2 \subseteq \Sigma^{*}$ be two languages and $\mathcal{R}: L_1 \to L_2$ a polynomial-time reduction from $L_1$ to $L_2$, i.e.\ $L_1 \leq_{\poly} L_2$. Furthermore, let $L_{B1}$ and $L_{B2}$ be bounded versions of $L_1$ and $L_2$, respectively.

If there is a strictly increasing polynomial $p: \mathbb{N} \to \mathbb{N}$ such that
\begin{equation} 
\label{eq:polynomial}
n_{\min, L_2}[\mathcal{R}(x)] \leq p\big(n_{\min, L_1}[x]\big)
\end{equation}
for every $x \in L$, then
\begin{equation}
\label{eq:bound-reduction}
\langle x, 1^n \rangle \mapsto \langle \mathcal{R}(x), 1^{p(n)}\rangle
\end{equation}
is a polynomial-time reduction from $L_{B1}$ to $L_{B2}$, hence $L_{B1} \leq_{\poly} L_{B2}$.
\end{theorem}
\begin{proof}
Since $\mathcal{R}$ and $p$ are polynomial-time maps, the map in Equation \eqref{eq:bound-reduction} is also polynomial-time. It remains to show that yes/no-instances are preserved via this map.
We have that $\langle x, 1^n \rangle \in L_{B1}$ if and only if $n \geq n_{\min, L_1}[x]$. This is equivalent to 
$$p(n) \geq p\big(n_{\min, L_1}[x]\big) \geq n_{\min, L_2}[\mathcal{R}(x)]$$
since $p$ is a strictly increasing function. But this is again equivalent to $\langle \mathcal{R}(x), 1^{p(n)} \rangle \in L_{B2}$.
\end{proof}

In words, Condition \eqref{eq:polynomial} demands that there is a polynomial that relates thresholds of $x$ and $\mathcal{R}(x)$ for all $x$. We require that $p$ is strictly increasing instead of mere non-decreasing as we need the equivalence of the statements $n \geq m$ and $p(n) \geq p(m)$ in the proof.

Many known reductions of undecidable problems implicitly contain such a polynomial $p$ in their construction. This gives an almost-for-free proof of the $\NP$-hardness of their bounded problems. However, most of these works do not make this polynomial explicit and therefore do not obtain the $\NP$-hardness results. While the theorem only assumes that $p(n_{\min,L_2}[x])$ upper bounds $n_{\min,L_1}[\mathcal{R}(x)]$, in all examples, we have an equality between these expressions. In \Cref{sec:examples}, we will present many examples of this behavior.

\Cref{thm:boundedHardness} also generalizes to other types of reductions. For example, we obtain an exponential-time reduction between the bounded versions when $\mathcal{R}$ is considered a exponential-time reduction and $p$ being a strictly increasing function that can be computed in exponential time.

\section{Halting problems as root problems}
\label{sec:root}

The result of \cref{thm:boundedHardness} gives only relative statements about hardness. Specifically, it allows to construct a reduction between bounded versions given a reduction between their original problems. To prove $\NP$/$\coNP$-hardness of bounded problems, we need root problems with bounded versions whose complexities are already known. In this section, we review two fundamental undecidable problems and their bounded versions, namely two variants of the halting problem. 

While $\HALT$ and $\BHALT$ are the most basic versions of halting problems, we need variations of the halting problem that take non-deterministic Turing machines as inputs. This is due to the fact that, while $\HALT$ is undecidable, $\BHALT$ is in $\P$. Since we want to prove $\NP$/$\coNP$-hardness of bounded problems, we need root problems with a $\NP$/$\coNP$-hard bounded version to start the reduction from. Therefore, we introduce two non-deterministic versions of $\HALT$, called $\NHALT$ and $\NHALTAll$, with an $\NP$-hard and a $\coNP$-hard bounded version, respectively.

\begin{enumerate}[label=(\alph*)]
	\item The problem $\NHALT$ checks the halting of a non-deterministic Turing machine on the empty tape. An instance is given by a description of a non-deterministic Turing machine $T$, which is accepted if and only if $T$ halts on the empty tape \footnote{In other words, it accepts if and only if there is a computation path such that $T$ halts along this path.}. Its bounded version $\BNHALT$ takes instances $\langle T, 1^n\rangle$ and accepts if and only if $T$ halts on the empty tape in at most $n$ steps. The unbounded problem is $\RE$-hard since it contains the (deterministic) halting problem on the empty tape, which is itself $\RE$-hard. Its bounded version $\BNHALT$ is $\NP$-hard.
	
\item The problem $\NHALTAll$ takes a description of a non-deterministic Turing machines $T$ as an instance, which is accepted if and only if $T$ halts on the empty tape along \emph{all} computation paths. Its bounded version $\BNHALTAll$ is given by instances $\langle T,1^n\rangle$ which are accepted if and only if $T$ halts on the empty tape within $n$ computational steps along \emph{all} computation paths. The unbounded problem is $\RE$-hard, and the bounded version is $\coNP$-hard.
\end{enumerate}

For more details on these problems and their complexity proofs, we refer to \Cref{app:HaltingProofs}. 

$\NHALT$ will be the root problem to prove the hardness of the bounded Post correspondence problem (\Cref{ssec:PCP})  and the bounded matrix mortality problem (\Cref{ssec:ZULCandMM}). 
$\NHALTAll$ will be the root problem to prove the hardness of the bounded Tiling problem (\Cref{ssec:TILING}).

While reductions for undecidable problems usually stem from the deterministic halting problem $\HALT$, here we need non-deterministic halting problems in order to prove $\NP$-hardness of the bounded versions. Canonical extensions of the reductions from $\HALT$ to a non-deterministic halting problem lead to different choices of root problems. For example, the Post correspondence problem has a similar structure as $\NHALT$, while the structure of the tiling problem relates to $\NHALTAll$. We will elaborate on these structures in the corresponding sections.

We expect that other variants of the halting problem serve as root problems for other complexity results;  
see \Cref{sec:questions} for further discussion.

\section{A tree of undecidable problems and their bounded versions}
\label{sec:examples}

In this section, we apply \Cref{thm:boundedHardness} to several undecidable problems in order to prove the $\NP$-hardness of the bounded versions.
The problems studied in this paper are summarized in \Cref{fig:ReductionTree}, where every edge corresponds to one application of the theorem. 
For a detailed treatment of these problems and further examples, we refer to \Cref{app:Examples}.

\hypersetup{linkcolor=gray}
\begin{figure}[thb]\centering
\begin{tikzpicture}[scale=0.88]

%\draw[fill=white, blur shadow={shadow blur steps=8,shadow blur extra rounding=1.3pt}] (1,5.7) rectangle (10.6,-2);
%\draw[fill=white] (1,5.7) rectangle (10.6,-2);

\draw[fill=color2!50!white, draw=none] (-1.5,-2.5) rectangle (4,4.5);

\node at (1.25,-1.2) {\footnotesize $\RE$-hard problems};
\node at (1.25, -1.6) {\footnotesize with $\NP$-hard};
\node at (1.25,-2){\footnotesize bounded version};

\draw[thick,black!50!white, -stealth] (0,0) -- (0.94,0.94) node[midway, yshift=-0.2cm, right] {\tiny Sec. \ref{ssec:PCP}};
\draw[thick,black!50!white, -stealth] (1,1) -- (1.94,1.94) node[midway, right, yshift=-0.2cm] {\tiny Sec. \ref{ssec:ZULCandMM}};
\draw[thick,black!50!white, -stealth] (2,2) -- (2.94,2.94) node[midway, right, yshift=-0.2cm] {\tiny Sec. \ref{ssec:ZULCandMM}};

\draw[thick,black!50!white, -stealth] (2,2) -- (1.06,2.94) node[midway, left, yshift=-0.2cm] {\tiny Sec. \ref{ssec:MPO}};
\draw[thick,black!50!white, -stealth] (1,3) -- (1.94,3.94) node[midway, right, yshift=-0.2cm] {\tiny Sec. \ref{ssec:POLY}};

\draw[thick,black!50!white, -stealth] (1,3) -- (0.06,3.94) node[midway, left, yshift=-0.2cm] {\tiny Sec. \ref{ssec:STAB}};
\draw[thick,black!50!white, -stealth] (1,1) -- (0.06,1.94) node[midway, left, yshift=-0.2cm] {\tiny Sec. \ref{ssec:REACH}};

\draw[fill=black, draw=color2!50!white, line width = 0.3mm] (0,0) circle (2pt) node[below right]{\footnotesize $\NHALT$};
\draw[fill=black, draw=color2!50!white, line width = 0.3mm] (1,1) circle (2pt) node[below right]{\footnotesize $\LangPCP$};
\draw[fill=black, draw=color2!50!white, line width = 0.3mm] (2,2) circle (2pt) node[below right]{\footnotesize $\ZULC$};
\draw[fill=black, draw=color2!50!white, line width = 0.3mm] (3,3) circle (2pt) node[below right]{\footnotesize $\MM$};
\draw[fill=black, draw=color2!50!white, line width = 0.3mm] (0,2) circle (2pt) node[below left]{\footnotesize $\REACH$};
\draw[fill=black, draw=color2!50!white, line width = 0.3mm] (1,3) circle (2pt) node[below left]{\footnotesize $\MPO$};
\draw[fill=black, draw=color2!50!white, line width = 0.3mm] (2,4) circle (2pt) node[below right]{\footnotesize $\POLY$};
\draw[fill=black, draw=color2!50!white, line width = 0.3mm] (0,4) circle (2pt) node[below left]{\footnotesize $\STAB$};

\begin{scope}[xshift=4.9cm]

\draw[fill=color3, draw=none] (-0.6,-2.5) rectangle (3.25,4.5);

\draw[thick,black!50!white, -stealth] (0,1) -- (0.94,1.94) node[midway, right, yshift=-0.2cm] {\tiny Sec. \ref{ssec:TILING}};
\draw[thick,black!50!white, -stealth] (1,2) -- (1.94,2.94) node[midway, right, yshift=-0.2cm] {\tiny Sec. \ref{ssec:GSE}};

\draw[fill=black, draw=color3, line width = 0.3mm] (0,1) circle (2pt) node[below right]{\footnotesize $\NHALTAll$};
\draw[fill=black, draw=color3, line width = 0.3mm] (1,2) circle (2pt) node[below right]{\footnotesize $\TILE$};
\draw[fill=black, draw=color3, line width = 0.3mm] (2,3) circle (2pt) node[below right]{\footnotesize $\GSE$};

\node at (1.375,-1.2) {\footnotesize $\RE$-hard problems};
\node at (1.375,-1.6) {\footnotesize with $\coNP$-hard};
\node at (1.375,-2) {\footnotesize bounded version};

\end{scope}

\end{tikzpicture}
\caption{The problems and reductions considered in this work. $\NHALT$ is the non-deterministic halting problem, $\LangPCP$ is the Post correspondence problem, $\REACH$ is the reachability problem for resource theories, $\ZULC$ is the zero in the upper left corner problem, $\MM$ is the matrix mortality problem, $\MPO$ is the positivity of Matrix product operators problem, $\STAB$ is the stability of positive maps problem and $\POLY$ is the polynomial positivity problem. $\NHALTAll$ is the non-deterministic halting problem on all computational paths, $\TILE$ is the Wang tiling problem, and $\GSE$ is the ground state energy problem.
$\NHALT$ and $\NHALTAll$ are the root problems, and every arrow corresponds to a reduction, explained in the referenced subsection.}
\label{fig:ReductionTree}
\end{figure}
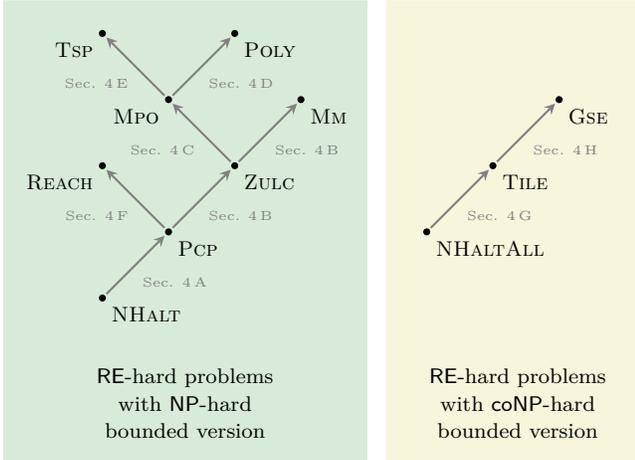
\hypersetup{linkcolor=black}

\subsection{The Post Correspondence Problem}
\label{ssec:PCP}

The Post correspondence problem ($\LangPCP$) \cite{Po46} is an undecidable problem with a particularly simple and intuitive formulation. For this reason, it is often used to prove undecidable results in quantum information theory \cite{Wo11}, including a version of the matrix product operator positivity problem \cite{Kl14}, threshold-problems for probabilistic and quantum finite automata \cite{Bl03}, or reachability problems in resource theories \cite{Sc21b}. It is stated as follows: 
\begin{quote}\emph{Given two finite sets of words, $\{a_1, \ldots, a_k\}$ and $\{b_1, \ldots, b_k\} \subseteq \Sigma^{*}$, is there a finite sequence of indices $i_1, \ldots, i_{\ell}$ such that
$$a_{i_1} a_{i_2} \ldots a_{i_{\ell}} = b_{i_1} b_{i_2} \ldots b_{i_{\ell}} \ ?$$}
\end{quote}
This decision problem is known to be $\RE$-complete via a reduction from the halting problem. Since $a_i$ and $b_i$ only appear in fixed pairs, this problem has an equivalent formulation in terms of dominoes 
$$d_i = \left[\frac{a_i}{b_i} \right].$$
The question is whether there exists a finite arrangement of dominoes that form a match, 
i.e.\ where the upper and lower parts coincide when the words are read across the dominoes (see \Cref{fig:PCPExample}).
\begin{figure}
\begin{tikzpicture}[scale=1.1]

\node at (1.3, 1.5) {Instance};

\node at (5.75,1.5) {Match};

\draw[shade, blur shadow={shadow blur steps=8,shadow blur extra rounding=1.3pt, shadow scale=0.96}] (0,0) rectangle (0.5,1);
\draw (0,0.5) -- (0.5,0.5);
\draw[thick, fill=color1] (0,0) rectangle (0.5,1);
\draw (0,0.5) -- (0.5,0.5);
\node at (0.25, 0.75) {\scriptsize $20$};
\node at (0.25, 0.25) {\scriptsize $0$};

\begin{scope}[xshift=0.7cm]
\draw[shade, blur shadow={shadow blur steps=8,shadow blur extra rounding=1.3pt, shadow scale=0.96}] (0,0) rectangle (0.5,1);
\draw[thick, fill=color2] (0,0) rectangle (0.5,1);
\draw (0,0.5) -- (0.5,0.5);
\node at (0.25, 0.75) {\scriptsize $0$};
\node at (0.25, 0.25) {\scriptsize $01$};
\end{scope}

\begin{scope}[xshift=1.4cm]
\draw[shade, blur shadow={shadow blur steps=8,shadow blur extra rounding=1.3pt, shadow scale=0.96}] (0,0) rectangle (0.5,1);
\draw[thick, fill=color3] (0,0) rectangle (0.5,1);
\draw (0,0.5) -- (0.5,0.5);
\node at (0.25, 0.75) {\scriptsize $012$};
\node at (0.25, 0.25) {\scriptsize $2$};
\end{scope}

\begin{scope}[xshift=2.1cm]
\draw[shade, blur shadow={shadow blur steps=8,shadow blur extra rounding=1.3pt, shadow scale=0.96}] (0,0) rectangle (0.5,1);
\draw[thick, fill=color3!10!white] (0,0) rectangle (0.5,1);
\draw (0,0.5) -- (0.5,0.5);
\node at (0.25, 0.75) {\scriptsize $1$};
\node at (0.25, 0.25) {\scriptsize $20$};
\end{scope}

\begin{scope}[xshift=4.4cm]
\draw[shade, blur shadow={shadow blur steps=8,shadow blur extra rounding=1.3pt, shadow scale=0.96}] (0,0) rectangle (0.5,1);
\draw[thick, fill=color2] (0,0) rectangle (0.5,1);
\draw (0,0.5) -- (0.5,0.5);
\node at (0.25, 0.75) {\scriptsize $0$};
\node at (0.25, 0.25) {\scriptsize $01$};
\end{scope}
\begin{scope}[xshift=4.95cm]
\draw[shade, blur shadow={shadow blur steps=8,shadow blur extra rounding=1.3pt, shadow scale=0.96}] (0,0) rectangle (0.5,1);
\draw[thick, fill=color3!10!white] (0,0) rectangle (0.5,1);
\draw (0,0.5) -- (0.5,0.5);
\node at (0.25, 0.75) {\scriptsize $1$};
\node at (0.25, 0.25) {\scriptsize $20$};
\end{scope}
\begin{scope}[xshift=5.5cm]
\draw[shade, blur shadow={shadow blur steps=8,shadow blur extra rounding=1.3pt, shadow scale=0.96}] (0,0) rectangle (0.5,1);
\draw[thick, fill=color1] (0,0) rectangle (0.5,1);
\draw (0,0.5) -- (0.5,0.5);
\node at (0.25, 0.75) {\scriptsize $20$};
\node at (0.25, 0.25) {\scriptsize $0$};
\end{scope}
\begin{scope}[xshift=6.05cm]
\draw[shade, blur shadow={shadow blur steps=8,shadow blur extra rounding=1.3pt, shadow scale=0.96}] (0,0) rectangle (0.5,1);
\draw[thick, fill=color2] (0,0) rectangle (0.5,1);
\draw (0,0.5) -- (0.5,0.5);
\node at (0.25, 0.75) {\scriptsize $0$};
\node at (0.25, 0.25) {\scriptsize $01$};
\end{scope}
\begin{scope}[xshift=6.6cm]
\draw[shade, blur shadow={shadow blur steps=8,shadow blur extra rounding=1.3pt, shadow scale=0.96}] (0,0) rectangle (0.5,1);
\draw[thick, fill=color3] (0,0) rectangle (0.5,1);
\draw (0,0.5) -- (0.5,0.5);
\node at (0.25, 0.75) {\scriptsize $012$};
\node at (0.25, 0.25) {\scriptsize $2$};
\end{scope}

\end{tikzpicture}
\caption{An instance of  $\LangPCP$ is a set of dominos (left).
This is a yes-instance if they form a match (right), i.e.\ the words on the top and the bottom coincide.}
\label{fig:PCPExample}
\end{figure}
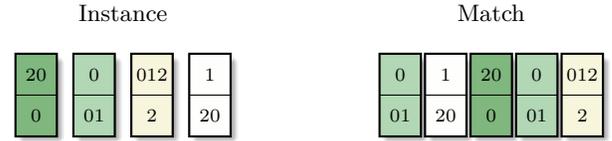

We define a bounded version of $\LangPCP$ that checks for sequences of length at most $n$: 
\begin{quote} \emph{Given a finite set of dominoes $\{d_1, \ldots, d_k\}$ and a number $n \in \mathbb{N}$ in unary, is there a matching arrangement of dominoes $d_{i_1}, \ldots, d_{i_\ell}$ with $\ell \leq n$?}
\end{quote}
This problem, denoted $\BPCP$, is a bounded version of $\LangPCP$ according to \Cref{def:boundedVersion}. It is known to be $\NP$-complete (see \cite{Ga79, Co74, Kl14} for the ideas of the reductions). The basic idea of the reduction is analogous to \Cref{thm:boundedHardness}, i.e.\ using the reduction of the (unbounded) undecidable problems to relate the bounding parameters via a polynomial-time map. Yet, the usual reductions do not directly give rise to a polynomial relation between the bounding parameters, contrary to what is claimed in \cite{Kl14}. 
We will now provide a reduction  $\NHALT \to \LangPCP$ leading to such a relation (and refer to  \Cref{app:PCP} for further details). 
Our approach is similar to that of \cite{Si06}.

\begin{figure*}[thb]\centering
\input{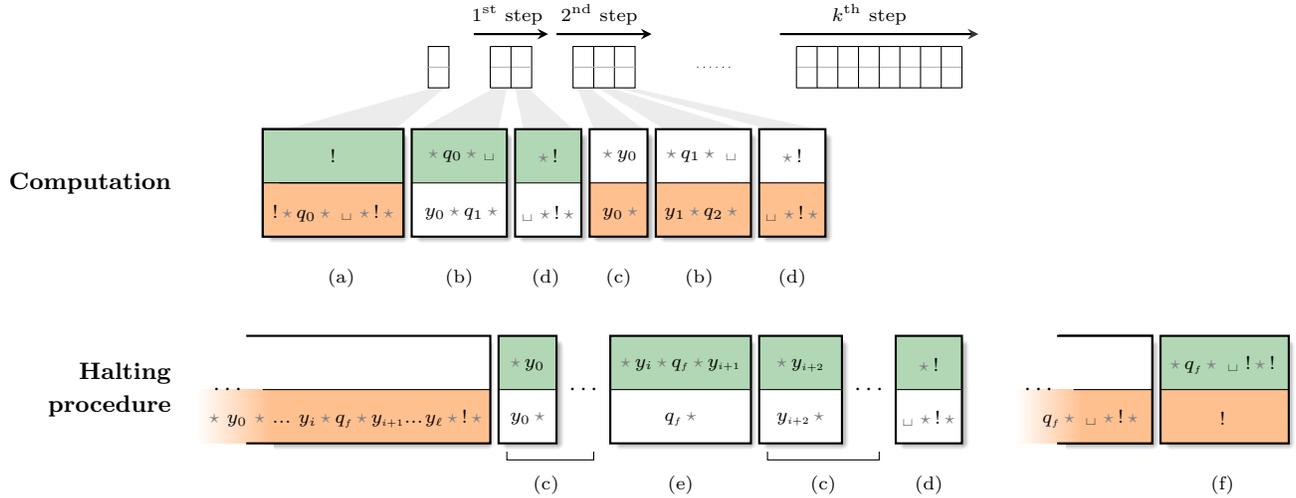}
\captionsetup{width=.8\linewidth, margin=0cm}
\caption{(Top) In the reduction  $\NHALT\to\LangPCP$,  
domino (a) contains the initial configuration of the TM, i.e.\ an empty tape with head at position zero. Each computation step is simulated by copying the lower string to the upper part in green. This is done by applying a transition domino (b), reproducing the tape (c), and adding a new empty tape slot (d). This generates a new string on the bottom, showing the new instantaneous description (white). Repeating the procedure simulates the computation. 
(Bottom) The halting of the Turing machine is mapped to the following match of tiles. 
When the Turing machine reaches the final state $q_f$, the instantaneous description is successively removed by dominoes (e). Adding a final domino (f) guarantees the match.}
\label{fig:PCPReduction}
\end{figure*}

We define a map $\mathcal{R}$ mapping a description of a Turing machine to a set of dominoes,  $\mathcal{R}(T) \coloneqq \langle d_1, \ldots, d_k\rangle$. 
This map mimics the description of $T$ (see \Cref{fig:PCPReduction}). 
For example, $d_1$ is a domino whose lower string is given by 
$$! \ \gstar \ q_0 \ \gstar \  \text{\textvisiblespace} \ \gstar \ ! \ \gstar$$ 
where $!$ and $\gstar$ are separator symbols,
 and $q_0$ and \textvisiblespace \ indicate that the Turing machine head is initially in state $q_0$ acting on an empty tape. 
 
The map $\mathcal{R}$ is a polynomial-time map; in particular, the number of dominoes $k$ is polynomial in the description size of $T$. From the construction of $\mathcal{R}$ it follows that $T$ halts on the empty tape if and only if there exists a match of dominoes $d_1, \ldots, d_k$. This implies that $\mathcal{R}$ is a polynomial-time reduction from the non-deterministic halting problem, which implies that $\LangPCP$ is $\RE$-hard. 

Refining this argument and using \Cref{thm:boundedHardness}, we obtain that $\mathcal{R}$ can be used as a reduction from $\BNHALT$ to $\BPCP$. 
Each computation step of $T$ on the empty tape is simulated by attaching dominoes, as shown in \Cref{fig:PCPReduction}. This procedure guarantees that $T$ halts within $n$ steps if and only if $d_1, \ldots, d_k$ form a match within 
$$p(n) \coloneqq (n+1) \cdot (n+2)$$ 
steps. Hence, the halting time of $T$ is polynomially related to the length of a minimal match of $\mathcal{R}(T)$. This proves that $\BPCP$ is $\NP$-hard by \Cref{thm:boundedHardness}.

Moreover, $\LangPCP$ is $\RE$-complete and $\BPCP$ is $\NP$-complete, 
by taking matching domino arrangements as certificates,
 and a polynomial-time verifier that checks arrangements.

\subsection{The Zero in the upper left corner and Matrix Mortality problem}
\label{ssec:ZULCandMM}

We now present the Matrix mortality problem ($\MM$) and the zero in the upper left corner problem ($\ZULC$) with their bounded versions. Both problems are undecidable and have been applied to prove the undecidability of quantum information problems such as the positivity of Matrix product operators \cite{De15} (see \Cref{ssec:MPO}), the reachability problem \cite{Wo11} (see \Cref{ssec:REACH}), or the measurement occurrence problem \cite{Ei12}. 

The matrix mortality problem is the following: 
\begin{quote} 
\emph{Given $A_1, \ldots, A_{k} \in \mathcal{M}_d(\mathbb{Q})$, is there a finite sequence $i_1, \ldots, i_\ell \in \{1,\ldots, k\}$ such that}
$$A_{i_1} \cdot A_{i_2} \cdots A_{i_\ell} = \textbf{0} \ ?$$
\end{quote}
Here, $\textbf{0}$ denotes the zero matrix, 
and $\mathcal{M}_d(\mathbb{Q})$ the set of $d\times d$ matrices over $\mathbb{Q}$. 
$\ZULC$ is almost identical to $\MM$, 
the only difference is that only the upper left corner of the product $A_{i_1} \cdot A_{i_2} \cdots A_{i_\ell}$ is asked to be zero. 
We define the bounded matrix mortality problem ($\BMM$) and the bounded zero in the upper left corner problem ($\BZULC$) by adding a parameter $n \in \mathbb{N}$ to every instance, 
 and asking whether the desired zeros can be realized within $n$ matrix multiplications.

The undecidability of $\MM$ was first proven by Paterson \cite{Pa70}. Since then, many tighter bounds on the number and size of matrices for both problems have been found (see \cite{Ca14} and references therein). It is also known that $\BMM$ is $\NP$-hard \cite{Bl97}. However, the proof relies on a reduction from the $\NP$-complete problem $\SAT$ and is therefore independent of the original reduction proving undecidability. To the best of our knowledge, the following is the first proof of the $\NP$-hardness of these bounded matrix problems using the same reductions as their unbounded versions.

We briefly sketch the reductions and refer to \Cref{app:MM} for details. Following \cite{Ha01b}, there exist polynomial-time reductions $\mathcal{R}: \LangPCP \to \ZULC$ and $\mathcal{Q}: \ZULC \to \MM$ with the following properties:
\begin{enumerate}[label=(\roman*)]
	\item The dominoes $\mathbf{d} \coloneqq \langle d_1, \ldots, d_k \rangle$ form a match of length $n$ if and only if the matrices $$\langle N_1, \ldots, N_{k'}\rangle \coloneqq \mathcal{R}(\mathbf{d})$$ multiply to a matrix with a zero in the upper left corner within $n$ matrix multiplications.
	\item The matrices $\mathbf{N} \coloneqq \langle N_1, \ldots, N_{\ell} \rangle$ form a zero in the upper left corner using $n$ matrix multiplications if and only if the matrices $$\langle M_1, \ldots, M_{\ell'} \rangle \coloneqq \mathcal{Q}(\mathbf{N})$$ multiply to a zero matrix within $n+2$ matrix multiplications.
\end{enumerate}

Together with \Cref{thm:boundedHardness}, these observations show that 
$$
\langle x, 1^n \rangle \mapsto \langle \mathcal{R}(x), 1^n \rangle
$$ 
is a polynomial-time reduction from $\BPCP$ to $\BZULC$, and 
$$
\langle x, 1^n \rangle \mapsto \langle \mathcal{Q}(x), 1^{n + 2} \rangle$$ 
is a polynomial-time reduction   from $\BZULC$ to $\BMM$. 
This proves that $\BZULC$ and $\BMM$ are $\NP$-hard.

Moreover, $\MM$ and $\ZULC$ are $\RE$-complete,   
and their bounded versions, $\BMM$ and $\BZULC$, are $\NP$-complete by taking matching matrix arrangements as certificates and a polynomial-time verifier checking the statements.

\subsection{The matrix product operator (MPO) positivity problem}
\label{ssec:MPO}

A matrix product operator (MPO) representation is a decomposition of a multipartite operator into local tensors according to a one-dimensional structure \cite{Ve04d, Zw04}. 
A local tensor $B$ defines a diagonal operator $\rho_n(B)$ for every system size $n$ (see \Cref{fig:MPO}). 
More precisely, given a family of $D \times D$ matrices $\left(B_{i}\right)$ for $i \in \{1, \ldots, d\}$, 
the diagonal (translationally invariant) MPO of size $n$ is given by
$$
\rho_n(B) \coloneqq \sum_{i_1, \ldots, i_n = 1}^{d} \tr\left(B_{i_1} \cdots B_{i_n}\right) \ket{i_1\ldots i_n} \bra{i_1 \ldots i_n}.
$$

\begin{figure}[th]\centering
\begin{tikzpicture}[scale=0.6]

\node at (-2,0.5) {$\rho_n(B) \coloneqq$};

\draw[thick,rounded corners=1.5mm] (0,0.5) -- (-0.5, 0.5) -- (-0.5, 1.2) -- (9,1.2) -- (9, 0.5) -- (8.5, 0.5);

\draw[thick, fill=color2, shade, blur shadow={shadow blur steps=8,shadow blur extra rounding=1.3pt, shadow scale=0.9}] (0,0) rectangle (1,1);
\draw[thick] (1,0.5) -- (1.5,0.5);
\draw[line width=1mm, draw=white] (0.5,1) -- (0.5,1.6);
\draw[thick] (0.5,1) -- (0.5,1.6);
\draw[thick] (0.5,0) -- (0.5,-0.6);
\draw[thick, fill=color2] (0,0) rectangle (1,1);
\node at (0.5, 0.5) {$B$};

\begin{scope}[xshift=1.5cm]
\draw[thick, fill=color2, shade, blur shadow={shadow blur steps=8,shadow blur extra rounding=1.3pt, shadow scale=0.9}] (0,0) rectangle (1,1);
\draw[thick] (1,0.5) -- (1.5,0.5);
\draw[line width=1mm, draw=white] (0.5,1) -- (0.5,1.6);
\draw[thick] (0.5,1) -- (0.5,1.6);
\draw[thick] (0.5,0) -- (0.5,-0.6);
\draw[thick, fill=color2] (0,0) rectangle (1,1);
\node at (0.5, 0.5) {$B$};
\end{scope}

\begin{scope}[xshift=3cm]
\draw[thick, fill=color2, shade, blur shadow={shadow blur steps=8,shadow blur extra rounding=1.3pt, shadow scale=0.9}] (0,0) rectangle (1,1);
\draw[thick] (1,0.5) -- (1.5,0.5);
\draw[line width=1mm, draw=white] (0.5,1) -- (0.5,1.6);
\draw[thick] (0.5,1) -- (0.5,1.6);
\draw[thick] (0.5,0) -- (0.5,-0.6);
\draw[thick, fill=color2] (0,0) rectangle (1,1);
\node at (0.5, 0.5) {$B$};
\end{scope}

\begin{scope}[xshift=4.5cm]
\draw[thick, fill=color2, shade, blur shadow={shadow blur steps=8,shadow blur extra rounding=1.3pt, shadow scale=0.9}] (0,0) rectangle (1,1);
\draw[thick] (1,0.5) -- (1.5,0.5);
\draw[line width=1mm, draw=white] (0.5,1) -- (0.5,1.6);
\draw[thick] (0.5,1) -- (0.5,1.6);
\draw[thick] (0.5,0) -- (0.5,-0.6);
\draw[thick, fill=color2] (0,0) rectangle (1,1);
\node at (0.5, 0.5) {$B$};

\draw[thick,dotted] (1.5,0.5) -- (2.6,0.5);

\end{scope}

\begin{scope}[xshift=7.5cm]
\draw[thick, fill=color2, shade, blur shadow={shadow blur steps=8,shadow blur extra rounding=1.3pt, shadow scale=0.9}] (0,0) rectangle (1,1);
\draw[line width=1mm, draw=white] (0.5,1) -- (0.5,1.6);
\draw[thick] (0.5,1) -- (0.5,1.6);
\draw[thick] (0.5,0) -- (0.5,-0.6);
\draw[thick, fill=color2] (0,0) rectangle (1,1);
\node at (0.5, 0.5) {$B$};

\draw[thick] (-0.5,0.5) -- (0,0.5);

\end{scope}

\draw[gray!70!white] (0.2, -0.8) -- (0.2,-1) -- node[midway,below]{\scriptsize $n$} (8.3,-1) -- (8.3, -0.8);

\end{tikzpicture}
\caption{Tensor network representation of the MPO $\rho_n(B)$. 
The MPO problem asks: Given a tensor $B$, is $\rho_n(B)$  positive semidefinite for all $n$?}
\label{fig:MPO}
\end{figure}
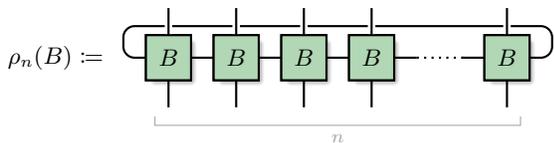

If  these MPO should represent density matrices, 
then $B$ should be such that $\rho_n(B)$ is positive semidefinite for every $n$. 
This property cannot be decided algorithmically, not even for classical states. 
In other words, the following $\MPO$ problem is undecidable: 
\begin{quote}
\emph{Given $B_1, \ldots, B_k \in \mathcal{M}_{D}(\mathbb{Q})$, is there $n \in \mathbb{N}$ such that $\rho_n(B) \ngeqslant 0$?}
\end{quote}

Note that an MPO is usually defined more generally; instead of restricting to families of diagonal (classical) matrices $B_i$, a general matrix product operator is defined via families of $D \times D$ matrices $(B_{i}^{j})$ for $i,j =1, \ldots, d$, addressing also non-diagonal entries of the matrix. However, as diagonal MPOs are contained in this definition, the undecidability of $\MPO$ as we defined it implies that the same problem for arbitrary matrix product operators is also undecidable.

Similar to previous bounded versions, we define $\BMPO$ by bounding the system size $n$: 
\begin{quote}
\emph{Given $B_1, \ldots, B_k$ and $n \in \mathbb{N}$, is there an $\ell \leq n$ such that $\rho_\ell(B) \ngeqslant 0$?}
\end{quote}
Note that $\MPO$ is usually stated in the negated way; yet, we use this definition for consistency with the definition of bounding. 

Let us now sketch the proof that $\BMPO$ is $\NP$-hard by using \Cref{thm:boundedHardness} together with the undecidability proof of \cite{De15}. For details, we refer to \Cref{app:MPO}. Following \cite{De15}, every instance $\langle A_1, \ldots, A_k\rangle$ of $\ZULC$ is mapped to $k+1$ matrices $B_{1}, \ldots, B_{k+1} \in \mathcal{M}_D(\mathbb{Z})$.
These matrices are constructed such that $(B_{i_1} \cdots B_{i_\ell})_{11} = 0$ if and only if $\exists j_1, \ldots, j_{\ell + 1} \in [k+1]$ such that
$$\tr\left(B_{j_{1}} \cdots B_{j_{\ell}} \cdot B_{j_{\ell + 1}}\right) < 0.$$
This implies that the family $A_1, \ldots, A_k$ generates a zero in the upper left corner using $n$ matrix multiplications if and only if $\rho_{n+1}(B) \ngeqslant 0$.
Setting $p(n) = n + 1$ proves that $\BMPO$ is $\NP$-hard by \Cref{thm:boundedHardness}.

Moreover, $\MPO$ is $\RE$-complete and $\BMPO$ is $\NP$-complete by defining negative diagonal entries as certificates.

While $\MPO$ precisely characterizes psd matrix product operators, in practice, algorithms distinguishing MPOs that are sufficiently positive or that violate positivity by at least an error $\varepsilon > 0$ are often acceptable. This is the idea behind weak membership problems. Along these lines, we define the approximate $\MPO$ problem $\MPO_{\varepsilon}$ as follows:

\begin{quote}
\emph{Given $C_1, \ldots, C_k \in \mathcal{M}_{D}(\mathbb{Q})$ with $\tr(\rho_\ell(C)) \leq 1$ for every $\ell \in \mathbb{N}$ and a family of errors $(\varepsilon_\ell)_{\ell \in \mathbb{N}}$ with $0 < \varepsilon_\ell \leq 1/\exp(\ell)$. Decide the following:
\begin{enumerate}[label=(\alph*)]
	\item Accept if $\exists n \in \mathbb{N}: \rho_n(C) \ngeqslant -\varepsilon_n \mathds{1}$.
	\item Reject if $\forall n \in \mathbb{N}: \rho_n(C) \geqslant \varepsilon_n \mathds{1}$.
\end{enumerate}}
\end{quote}

$\MPO_\varepsilon$ is undecidable using the same reduction as above and the fact that $\tr(\rho_n(C))$ increases exponentially in $n$ in the above reduction. Following the usual bounding process, we define $\BMPO_{\varepsilon}$ by bounding $n$:

\begin{quote}\emph{Given $C_1, \ldots, C_k \in \mathcal{M}_{D}(\mathbb{Q})$ with $\tr(\rho_\ell(C)) \leq 1$ for every $\ell \in \mathbb{N}$, a family of errors $(\varepsilon_\ell)_{\ell \in \mathbb{N}}$ with $0 < \varepsilon_\ell \leq 1/\exp(\ell)$ and $n \in \mathbb{N}$. Decide the following:
\begin{enumerate}[label=(\alph*)]
	\item Accept if $\exists \ell \leq n: \rho_\ell(C) \ngeqslant -\varepsilon_n \mathds{1}$.
	\item Reject if $\forall \ell \leq n: \rho_\ell(C) \geqslant \varepsilon_n \mathds{1}$.
\end{enumerate}}
\end{quote}

It follows that $\BMPO_{\varepsilon}$ is a bounded version of $\MPO_{\varepsilon}$ according to \Cref{def:boundedVersion}. Moreover, \Cref{thm:boundedHardness} implies that $\BMPO_{\varepsilon}$ is also $\NP$-hard.

We remark that Kliesch et al.\ \cite{Kl14} present a similar idea, by constructing a reduction from $\LangPCP$ to an alternative version of $\MPO$ and bounding both problems.

\subsection{The polynomial positivity problem}
\label{ssec:POLY}

The undecidability of $\MPO$ leads to the undecidability of other positivity problems. 
One of them concerns deciding the positivity of a certain class of polynomials \cite{De21b}: 
\begin{quote}
\emph{Given a family of polynomials $q_{\alpha, \beta}(\underline{x})$ for $\alpha, \beta \in \{1, \ldots, D\}$ with integer coefficients, is there an $n \in \mathbb{N}$ such that the polynomial
\begin{equation}\label{eq:pn}
p_n(\underline{x}_1, \ldots, \underline{x}_n) \coloneqq 
\sum_{\alpha_1, \ldots, \alpha_n = 1}^{D} 
q_{\alpha_1, \alpha_2}(\underline{x}_1) \cdots q_{\alpha_n, \alpha_1}(\underline{x}_n)
\end{equation}
is not nonnegative (i.e.\ $p_n(\mathbf{a}) < 0$ for some $\mathbf{a} \in \mathbb{R}^{d \cdot n}$)?}
\end{quote}
Here $\underline{x}_i$ denotes a $d$-tuple of variables, for every $i$. We define this problem as $\POLY$ and its bounded version (by restricting to checking nonnegativity of $p_k$ for $k \leq n$) by $\BPOLY$.

Following the proof of \cite{De21b}, $\POLY$ is $\RE$-hard since there exists a polynomial-time map 
$$
\mathcal{R}(\langle B_1, \ldots, B_k\rangle) \coloneqq \left\langle q_{\alpha,\beta}: \alpha, \beta = 1, \ldots, D\right\rangle$$
such that
$$
\rho_n(B) \geqslant 0 \text{ if and only if } p_n \text{ is nonnegative.}
$$
This implies that $\langle B, 1^n\rangle \mapsto \langle \mathcal{R}(B), 1^n\rangle$ defines a reduction from $\BMPO$ to $\BPOLY$. 
It follows that $\BPOLY$ is $\NP$-hard. 
We refer to \Cref{app:POLY} for further details.

\subsection{Stability of positive maps}
\label{ssec:STAB}

Another undecidable problem related to positivity concerns tensor products of positive maps. 
A map 
$$
\mathcal{P}: \mathcal{M}_d(\mathbb{C}) \to \mathcal{M}_d(\mathbb{C})
$$
is called \emph{positive} if it maps positive semidefinite  matrices to positive semidefinite matrices. Such a map is called \emph{$n$-tensor-stable positive} if $\mathcal{P}^{\otimes n}$ is a positive map, and \emph{tensor-stable positive} if it is $n$-tensor-stable positive for all $n \in \mathbb{N}$. The existence of non-trivial tensor-stable positive maps relates to the existence of NPT bound-entangled states \cite{Mu15}.

Let us define the   $n$-fold Matrix Multiplication tensor as 
$$
\ket{\chi_n} \coloneqq \sum_{i_1, \ldots, i_n = 1}^{s} \ket{i_1, i_2} \otimes \ket{i_2, i_3} \otimes \cdots \otimes \ket{i_n, i_1}
$$
and denote the projection to this vector by 
\begin{equation}\label{eq:chin}
\chi_n \coloneqq \ket{\chi_n} \bra{\chi_n}. 
\end{equation}
The following problem is undecidable \cite{Ey21}: 
\begin{quote}
\emph{Given a positive map $\mathcal{P}: \mathcal{M}_d(\mathbb{C}) \to \mathcal{M}_d(\mathbb{C})$, is $\mathcal{P}^{\otimes n}(\chi_n)$ not positive semidefinite for some $n \in \mathbb{N}$?}
\end{quote}
We denote this problem by $\STAB$. 
Its bounded version, $\BSTAB$ takes instances $\langle \mathcal{P}, 1^n\rangle$ and asks the same question for $k$-fold tensor products with $k \leq n$. 

From the proof of \cite{Ey21} it follows that there is a polynomial-time map 
$ \langle B_1, \ldots, B_k \rangle \mapsto \mathcal{P}$ 
such that the resulting $\mathcal{P}$ satisfies that $\mathcal{P}^{\otimes n}(\chi_n) = \rho_n(B)$.
This shows that $\STAB$ is $\RE$-hard by a reduction from $\MPO$. 
Moreover, since
$$
\rho_n(B) \geqslant 0 \quad \text{ if and only if } \quad \mathcal{P}^{\otimes n}(\chi_n) \geqslant 0,
$$ 
it follows $\BSTAB$ is $\NP$-hard by applying \Cref{thm:boundedHardness} together with the fact that $\BMPO$ is $\NP$-hard. We refer to \Cref{app:STAB} for more details on the reduction.

\subsection{The reachability problem in quantum information}
\label{ssec:REACH}

The reachability problem in quantum information concerns the question whether a resource state $\rho$ (given as a density matrix) can be converted to another state $\sigma$ by using only free resource operations from a fixed set $\mathcal{F} \coloneqq \{\Phi_1, \ldots, \Phi_k\}$. More precisely, we define $\REACH$ as follows:
\begin{quote}
\emph{Given density matrices $\rho$, $\sigma \in \mathcal{M}_d(\mathbb{C})$ and a set $\mathcal{F}$ of free operations $\mathcal{M}_d(\mathbb{C}) \to \mathcal{M}_d(\mathbb{C})$, is there a map 
$$\Phi \coloneqq \Phi_{i_n} \circ \Phi_{i_{n-1}} \circ \cdots \circ \Phi_{i_1}$$
in the free semigroup $\mathcal{F}^{*}$ such that $\sigma = \Phi(\rho)$?}
\end{quote}
The free semigroup $\mathcal{F}^{*}$ of $\mathcal{F}$ consists of all maps generated by finite compositions of maps in $\mathcal{F}$. We denote by $\mathcal{F}^{n}$ the set of all operations arising from at most $n$ compositions of maps in $\mathcal{F}$, and define the bounded version $\BREACH$ by replacing $\mathcal{F}^{*}$ with $\mathcal{F}^n$ in the above problem statement.

$\REACH$ is undecidable via a reduction from $\LangPCP$ \cite{Sc21b}. 
We now prove that the bounded version $\BREACH$ is $\NP$-hard. 
We rely on Scandi and Surace's work \cite{Sc21b}, who provide a polynomial-time reduction $\mathcal{R}$ mapping dominoes $d_{i}$ to two types of resource maps $H_i^{\lambda}, G_i^{\lambda}$ for $\lambda \in (0,1)$. The set of free resource operations is then specified by
$$\mathcal{F} = \big\{\mathds{1}, H^{\lambda}_i, G_i^{\lambda} : i=1, \ldots, r \text{ and } \lambda \in (0,1)\big\}.$$
For a state $\rho \in \mathcal{M}_4(\mathbb{C})$, it is shown that 
$$\sigma \coloneqq \lambda \rho + (1 - \lambda)\frac{\mathds{1}}{4}$$ 
is reachable via operations in $\mathcal{F}^{*}$ if and only if there exists a match of the corresponding dominoes in $\LangPCP$. This shows that $\REACH$ is $\RE$-hard. More specifically, there exists a match of length $n$ if and only if 
$$\sigma = G^{\lambda_n}_{i_n} \circ \cdots \circ G^{\lambda_1}_{i_1} \circ H^{\lambda_1}_{i_1} \circ \cdots \circ H^{\lambda_n}_{i_n}(\rho)$$
for a choice $\lambda_1, \ldots, \lambda_n \in (0,1)$. In other words, a threshold parameter $n$ in $\BPCP$ is mapped to a threshold $2n$ in $\BREACH$. This proves that $\BREACH$ is $\NP$-hard by applying \Cref{thm:boundedHardness}.

\subsection{The tiling problem}
\label{ssec:TILING}

Let us now consider the Wang tiling problem. This problem has been used to prove undecidability in many physics-related problems, like the spectral gap problem in 2D \cite{Cu15}, 2D PEPS problems \cite{Sc20}, or the universality of translational invariant, classical spin Hamiltonians in 2D \cite{Ko19}.

\begin{figure}[thb]\centering
\begin{tikzpicture}[scale=0.6]

\node at (2.25,3.5) {Instance};
\node at (10.75,3.5) {Valid tiling};

\draw[thick, shade, blur shadow={shadow blur steps=8,shadow blur extra rounding=1.3pt, shadow scale=0.9}] (0,0) rectangle (1,1);
\draw[fill=color2] (0,0) -- (1,0) -- (0.5,0.5) -- cycle;
\draw[fill=color3] (0,0) -- (0,1) -- (0.5,0.5) -- cycle;
\draw[fill=color4] (1,0) -- (1,1) -- (0.5,0.5) -- cycle;
\draw[fill=color1] (0,1) -- (1,1) -- (0.5,0.5) -- cycle;
\draw (0,0) -- (1,1);
\draw (0,1) -- (1,0);
\draw[thick] (0,0) rectangle (1,1);

\begin{scope}[xshift=1.25cm]
\draw[thick, shade, blur shadow={shadow blur steps=8,shadow blur extra rounding=1.3pt, shadow scale=0.9}] (0,0) rectangle (1,1);
\draw[fill=color1] (0,0) -- (1,0) -- (0.5,0.5) -- cycle;
\draw[fill=color4] (0,0) -- (0,1) -- (0.5,0.5) -- cycle;
\draw[fill=color4] (1,0) -- (1,1) -- (0.5,0.5) -- cycle;
\draw[fill=color1] (0,1) -- (1,1) -- (0.5,0.5) -- cycle;
\draw (0,0) -- (1,1);
\draw (0,1) -- (1,0);
\draw[thick] (0,0) rectangle (1,1);
\end{scope}

\begin{scope}[xshift=2.5cm]
\draw[thick, shade, blur shadow={shadow blur steps=8,shadow blur extra rounding=1.3pt, shadow scale=0.9}] (0,0) rectangle (1,1);
\draw[fill=color1] (0,0) -- (1,0) -- (0.5,0.5) -- cycle;
\draw[fill=color2] (0,0) -- (0,1) -- (0.5,0.5) -- cycle;
\draw[fill=color3] (1,0) -- (1,1) -- (0.5,0.5) -- cycle;
\draw[fill=color4] (0,1) -- (1,1) -- (0.5,0.5) -- cycle;
\draw (0,0) -- (1,1);
\draw (0,1) -- (1,0);
\draw[thick] (0,0) rectangle (1,1);
\end{scope}

\begin{scope}[xshift=3.75cm]
\draw[thick, shade, blur shadow={shadow blur steps=8,shadow blur extra rounding=1.3pt, shadow scale=0.9}] (0,0) rectangle (1,1);
\draw[fill=color4] (0,0) -- (1,0) -- (0.5,0.5) -- cycle;
\draw[fill=color2] (0,0) -- (0,1) -- (0.5,0.5) -- cycle;
\draw[fill=color4] (1,0) -- (1,1) -- (0.5,0.5) -- cycle;
\draw[fill=color3] (0,1) -- (1,1) -- (0.5,0.5) -- cycle;
\draw (0,0) -- (1,1);
\draw (0,1) -- (1,0);
\draw[thick] (0,0) rectangle (1,1);
\end{scope}

\begin{scope}[yshift=1.25cm, xshift=0.5cm]
\draw[thick, shade, blur shadow={shadow blur steps=8,shadow blur extra rounding=1.3pt, shadow scale=0.9}] (0,0) rectangle (1,1);
\draw[fill=color3] (0,0) -- (1,0) -- (0.5,0.5) -- cycle;
\draw[fill=color4] (0,0) -- (0,1) -- (0.5,0.5) -- cycle;
\draw[fill=color1] (1,0) -- (1,1) -- (0.5,0.5) -- cycle;
\draw[fill=color1] (0,1) -- (1,1) -- (0.5,0.5) -- cycle;
\draw (0,0) -- (1,1);
\draw (0,1) -- (1,0);
\draw[thick] (0,0) rectangle (1,1);

\begin{scope}[xshift=1.25cm]
\draw[thick, shade, blur shadow={shadow blur steps=8,shadow blur extra rounding=1.3pt, shadow scale=0.9}] (0,0) rectangle (1,1);
\draw[fill=color3] (0,0) -- (1,0) -- (0.5,0.5) -- cycle;
\draw[fill=color4] (0,0) -- (0,1) -- (0.5,0.5) -- cycle;
\draw[fill=color2] (1,0) -- (1,1) -- (0.5,0.5) -- cycle;
\draw[fill=color3] (0,1) -- (1,1) -- (0.5,0.5) -- cycle;
\draw (0,0) -- (1,1);
\draw (0,1) -- (1,0);
\draw[thick] (0,0) rectangle (1,1);
\end{scope}

\begin{scope}[xshift=2.5cm]
\draw[thick, shade, blur shadow={shadow blur steps=8,shadow blur extra rounding=1.3pt, shadow scale=0.9}] (0,0) rectangle (1,1);
\draw[fill=color2] (0,0) -- (1,0) -- (0.5,0.5) -- cycle;
\draw[fill=color1] (0,0) -- (0,1) -- (0.5,0.5) -- cycle;
\draw[fill=color4] (1,0) -- (1,1) -- (0.5,0.5) -- cycle;
\draw[fill=color2] (0,1) -- (1,1) -- (0.5,0.5) -- cycle;
\draw (0,0) -- (1,1);
\draw (0,1) -- (1,0);
\draw[thick] (0,0) rectangle (1,1);
\end{scope}
\end{scope}

\begin{scope}[xshift=10cm, yshift=0.5cm, scale=0.8]

\draw[thick, shade, blur shadow={shadow blur steps=8,shadow blur extra rounding=1.3pt, shadow scale=0.96}] (-1,-1) -- (3,-1) -- (3,2) -- (2,2) -- (2,3) -- (0,3) -- (0,2) -- (-1,2) -- cycle;

\begin{scope}[xshift=0cm]
\draw[fill=color1] (0,0) -- (1,0) -- (0.5,0.5) -- cycle;
\draw[fill=color2] (0,0) -- (0,1) -- (0.5,0.5) -- cycle;
\draw[fill=color3] (1,0) -- (1,1) -- (0.5,0.5) -- cycle;
\draw[fill=color4] (0,1) -- (1,1) -- (0.5,0.5) -- cycle;
\draw (0,0) -- (1,1);
\draw (0,1) -- (1,0);
\draw[thick] (0,0) rectangle (1,1);
\end{scope}

\begin{scope}[xshift=1cm]
\draw[fill=color2] (0,0) -- (1,0) -- (0.5,0.5) -- cycle;
\draw[fill=color3] (0,0) -- (0,1) -- (0.5,0.5) -- cycle;
\draw[fill=color4] (1,0) -- (1,1) -- (0.5,0.5) -- cycle;
\draw[fill=color1] (0,1) -- (1,1) -- (0.5,0.5) -- cycle;
\draw (0,0) -- (1,1);
\draw (0,1) -- (1,0);
\draw[thick] (0,0) rectangle (1,1);
\end{scope}

\begin{scope}[xshift=2cm]
\draw[fill=color1] (0,0) -- (1,0) -- (0.5,0.5) -- cycle;
\draw[fill=color4] (0,0) -- (0,1) -- (0.5,0.5) -- cycle;
\draw[fill=color4] (1,0) -- (1,1) -- (0.5,0.5) -- cycle;
\draw[fill=color1] (0,1) -- (1,1) -- (0.5,0.5) -- cycle;
\draw (0,0) -- (1,1);
\draw (0,1) -- (1,0);
\draw[thick] (0,0) rectangle (1,1);
\end{scope}

\begin{scope}[yshift=1cm]
\draw[fill=color4] (0,0) -- (1,0) -- (0.5,0.5) -- cycle;
\draw[fill=color2] (0,0) -- (0,1) -- (0.5,0.5) -- cycle;
\draw[fill=color4] (1,0) -- (1,1) -- (0.5,0.5) -- cycle;
\draw[fill=color3] (0,1) -- (1,1) -- (0.5,0.5) -- cycle;
\draw (0,0) -- (1,1);
\draw (0,1) -- (1,0);
\draw[thick] (0,0) rectangle (1,1);
\end{scope}

\begin{scope}[yshift=1cm, xshift=1cm]
\draw[fill=color1] (0,0) -- (1,0) -- (0.5,0.5) -- cycle;
\draw[fill=color4] (0,0) -- (0,1) -- (0.5,0.5) -- cycle;
\draw[fill=color4] (1,0) -- (1,1) -- (0.5,0.5) -- cycle;
\draw[fill=color1] (0,1) -- (1,1) -- (0.5,0.5) -- cycle;
\draw (0,0) -- (1,1);
\draw (0,1) -- (1,0);
\draw[thick] (0,0) rectangle (1,1);
\end{scope}

\begin{scope}[yshift=1cm, xshift=2cm]
\draw[fill=color1] (0,0) -- (1,0) -- (0.5,0.5) -- cycle;
\draw[fill=color4] (0,0) -- (0,1) -- (0.5,0.5) -- cycle;
\draw[fill=color4] (1,0) -- (1,1) -- (0.5,0.5) -- cycle;
\draw[fill=color1] (0,1) -- (1,1) -- (0.5,0.5) -- cycle;
\draw (0,0) -- (1,1);
\draw (0,1) -- (1,0);
\draw[thick] (0,0) rectangle (1,1);
\end{scope}

\begin{scope}[yshift=-1cm]
\draw[fill=color3] (0,0) -- (1,0) -- (0.5,0.5) -- cycle;
\draw[fill=color4] (0,0) -- (0,1) -- (0.5,0.5) -- cycle;
\draw[fill=color1] (1,0) -- (1,1) -- (0.5,0.5) -- cycle;
\draw[fill=color1] (0,1) -- (1,1) -- (0.5,0.5) -- cycle;
\draw (0,0) -- (1,1);
\draw (0,1) -- (1,0);
\draw[thick] (0,0) rectangle (1,1);
\end{scope}

\begin{scope}[yshift=-1cm,xshift=1cm]
\draw[fill=color2] (0,0) -- (1,0) -- (0.5,0.5) -- cycle;
\draw[fill=color1] (0,0) -- (0,1) -- (0.5,0.5) -- cycle;
\draw[fill=color4] (1,0) -- (1,1) -- (0.5,0.5) -- cycle;
\draw[fill=color2] (0,1) -- (1,1) -- (0.5,0.5) -- cycle;
\draw (0,0) -- (1,1);
\draw (0,1) -- (1,0);
\draw[thick] (0,0) rectangle (1,1);
\end{scope}

\begin{scope}[yshift=-1cm,xshift=2cm]
\draw[fill=color1] (0,0) -- (1,0) -- (0.5,0.5) -- cycle;
\draw[fill=color4] (0,0) -- (0,1) -- (0.5,0.5) -- cycle;
\draw[fill=color4] (1,0) -- (1,1) -- (0.5,0.5) -- cycle;
\draw[fill=color1] (0,1) -- (1,1) -- (0.5,0.5) -- cycle;
\draw (0,0) -- (1,1);
\draw (0,1) -- (1,0);
\draw[thick] (0,0) rectangle (1,1);
\end{scope}

\begin{scope}[yshift=2cm]
\draw[fill=color3] (0,0) -- (1,0) -- (0.5,0.5) -- cycle;
\draw[fill=color4] (0,0) -- (0,1) -- (0.5,0.5) -- cycle;
\draw[fill=color2] (1,0) -- (1,1) -- (0.5,0.5) -- cycle;
\draw[fill=color3] (0,1) -- (1,1) -- (0.5,0.5) -- cycle;
\draw (0,0) -- (1,1);
\draw (0,1) -- (1,0);
\draw[thick] (0,0) rectangle (1,1);
\end{scope}

\begin{scope}[yshift=2cm,xshift=1cm]
\draw[fill=color1] (0,0) -- (1,0) -- (0.5,0.5) -- cycle;
\draw[fill=color2] (0,0) -- (0,1) -- (0.5,0.5) -- cycle;
\draw[fill=color3] (1,0) -- (1,1) -- (0.5,0.5) -- cycle;
\draw[fill=color4] (0,1) -- (1,1) -- (0.5,0.5) -- cycle;
\draw (0,0) -- (1,1);
\draw (0,1) -- (1,0);
\draw[thick] (0,0) rectangle (1,1);
\end{scope}

\begin{scope}[yshift=-1cm,xshift=-1cm]
\draw[fill=color4] (0,0) -- (1,0) -- (0.5,0.5) -- cycle;
\draw[fill=color2] (0,0) -- (0,1) -- (0.5,0.5) -- cycle;
\draw[fill=color4] (1,0) -- (1,1) -- (0.5,0.5) -- cycle;
\draw[fill=color3] (0,1) -- (1,1) -- (0.5,0.5) -- cycle;
\draw (0,0) -- (1,1);
\draw (0,1) -- (1,0);
\draw[thick] (0,0) rectangle (1,1);
\end{scope}

\begin{scope}[xshift=-1cm]
\draw[fill=color3] (0,0) -- (1,0) -- (0.5,0.5) -- cycle;
\draw[fill=color4] (0,0) -- (0,1) -- (0.5,0.5) -- cycle;
\draw[fill=color2] (1,0) -- (1,1) -- (0.5,0.5) -- cycle;
\draw[fill=color3] (0,1) -- (1,1) -- (0.5,0.5) -- cycle;
\draw (0,0) -- (1,1);
\draw (0,1) -- (1,0);
\draw[thick] (0,0) rectangle (1,1);
\end{scope}

\begin{scope}[yshift=1cm, xshift=-1cm]
\draw[fill=color3] (0,0) -- (1,0) -- (0.5,0.5) -- cycle;
\draw[fill=color4] (0,0) -- (0,1) -- (0.5,0.5) -- cycle;
\draw[fill=color2] (1,0) -- (1,1) -- (0.5,0.5) -- cycle;
\draw[fill=color3] (0,1) -- (1,1) -- (0.5,0.5) -- cycle;
\draw (0,0) -- (1,1);
\draw (0,1) -- (1,0);
\draw[thick] (0,0) rectangle (1,1);
\end{scope}

\end{scope}

\end{tikzpicture}
\caption{An instance of $\TILE$ is a set of tiles (left). A set of tiles is a yes-instance if there exists a valid tiling of the plane. Part of a potentially valid tiling is shown on the right. In a valid tiling, the colors of adjacent tiles must coincide and the tiles cannot be rotated.}
\end{figure}
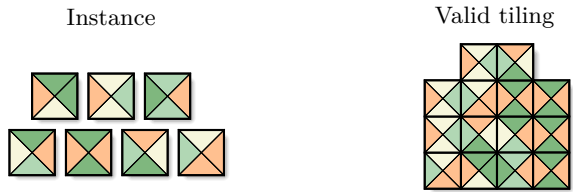

A tile is given by a square with different colors on each side of the tile (see \Cref{fig:Tiling}). Given a finite set of tiles, a valid tiling is an arrangement of tiles whose adjacent edges coincide. Moreover, all tiles have a fixed orientation, i.e.\ they cannot rotate. We study the following variant: 
\begin{quote}
\emph{Given a set of tiles $\mathcal{T} = \{t_1, \ldots, t_k\}$, is it impossible to tile the plane when $t_1$ is in the origin?}
\end{quote}
Note that this problem is usually stated in the negated form, but this formulation is more convenient for our purposes. 
The constraint on the fixed tile in the origin can also be removed \cite{Be66, Ro71}; we stick to this version for simplicity. 
The corresponding bounded version is the following:
\begin{quote}
\emph{Given a set of tiles $\mathcal{T} = \{t_1, \ldots, t_k\}$ and $n \in \mathbb{N}$, is it impossible to tile $\mathbb{Z}_n^2$ when $t_1$ is in the origin?}
\end{quote}
Here we denote by $\mathbb{Z}_n^2 \coloneqq \{-n, \ldots, 0, \ldots, n\}^2$ the square grid of size $(2n + 1) \times (2n + 1)$ around the origin.

Let us now sketch the proof that $\TILE$ is $\RE$-hard and that $\BTILE$ is $\coNP$-hard.   
This will imply that the tiling problem in its usual formulation (``can the plane be tiled?'') is $\coRE$-hard and its bounded version is $\NP$-hard. We refer to \Cref{app:Tiling} for details. 

In contrast to the previous examples, we now construct a reduction from $\NHALTAll$ instead of $\NHALT$. While to check whether $\{d_1, \ldots, d_k\}$ is a yes-instance of $\BPCP$, one needs to find a \emph{single} matching arrangement, to verify whether $\{t_1, \ldots, t_k\}$ is a yes-instance of $\BTILE$ one has to check (for a fixed size $n$) whether \emph{all} arrangements of tiles in $\mathbb{Z}_n^2$ are invalid. This structure is similar to $\NHALTAll$, where for a fixed computation time $n$, one needs to check whether a given Turing machine $T$ halts on \emph{all} computation steps. More precisely, there is a polynomial relation between the bounding parameters of $\BTILE$ and $\BNHALTAll$, as needed for \Cref{thm:boundedHardness}.

We build a polynomial-time reduction from  $\NHALTAll $ to $\TILE$ following \cite{Ro71}. 
The reduction maps a description of a Turing machine $T$ to a set of tiles representing either a slot in the tape or a computational step. 
The (infinite) starting tape is mapped to the fixed origin tile representing the empty tape with head position at zero. 
Filling up a new line corresponds to one computational step. 
This reduction also applies to non-deterministic Turing machines. 

The reduction is such that the tiling cannot be continued after filling up $n$ lines if and only if $T$ halts on all computation paths after at most $n$ computational steps. We refer to \Cref{fig:Tiling}  and \Cref{app:Tiling} for further details on the reduction. This proves that $\TILE$ is undecidable. 
By  \Cref{thm:boundedHardness}, we obtain that $\BTILE$ is $\coNP$-hard, since the maximal halting time $n$ on every computation path is mapped to the termination size $n+1$. 

In addition, $\TILE$ is $\RE$-complete by taking a system size where all tilings terminate as a certificate and an exponential-time verifier checking all tilings of this size. $\BTILE$ is $\coNP$-complete by choosing tilings as a certificate and a polynomial-time verifier checking the validity of the tiling. This highlights that when proving completeness, \emph{not} every construction in the unbounded case trivially translates to the bounded version.

\begin{figure}[thb]\centering
\begin{tikzpicture}

\draw (-3.5,0) -- (2.5,0);
\draw (-3.5,1) -- (2.5,1);
\draw (-3.5,2) -- (2.5,2);
\draw (-3.5,3) -- (2.5,3);
\draw (-3.5,4) -- (2.5,4);

\draw (-3,-0.5) -- (-3,4.5);
\draw (-2,-0.5) -- (-2,4.5); 
\draw (-1,-0.5) -- (-1,4.5); 
\draw (0,-0.5) -- (0,4.5); 
\draw (1,-0.5) -- (1,4.5); 
\draw (2,-0.5) -- (2,4.5);

\draw[draw=none, fill=color1] (-1, 2) -- (0,2) -- (0,1) -- (-1,1) -- cycle;
\draw[very thick] (-1,2) -- (-0.8,2);
\draw[very thick] (-0.2, 2) -- (0,2) -- (0,1) -- (-1,1) -- (-1,2);
\draw[draw=none, fill=color4] (-1,2) -- (0,2) -- (0,3) -- (-1,3) -- cycle;
\draw[very thick] (-0.2, 2) -- (0,2) -- (0,3) -- (-1,3) -- (-1,2) -- (-0.8,2);

\draw[draw=white, thick, fill=color4] (0,2) -- (1,2) -- (1,3) -- (0,3) -- cycle;
\draw[very thick, fill=color4] (0.7,2) -- (1,2) -- (1,3) -- (0,3) -- (0,2) -- (0.3,2);
\draw[very thick, fill=color4] (-1,3) -- (0,3) -- (0,4) -- (-1,4) -- cycle;
\draw[very thick, fill=color4] (0,3) -- (1,3) -- (1,4) -- (0,4) -- cycle;
\draw[draw=none, fill=color4] (-1,4) -- (1,4) -- (1,4.5) -- (-1,4.5) -- cycle;

\draw[very thick] (1,4.5)-- (1,4) -- (0,4) -- (0,4.5);
\draw[very thick] (-1,4.5)-- (-1,4) -- (0,4) -- (0,4.5);
%\draw[very thick] (1,4) -- (1,4.5);
%\draw[very thick] (-1,4) -- (-1,4.5);

\draw[stealth-stealth] (-1,1.5) -- (0,1.5);
\draw[-stealth,gray!50!white] (0,1.5) -- (1,1.5);
\draw[-stealth,gray!50!white] (1,1.5) -- (2,1.5);
\draw[stealth-,gray!50!white] (-3,1.5) -- (-2,1.5);
\draw[stealth-,gray!50!white] (-2,1.5) -- (-1,1.5);

\draw[-stealth, gray!50!white] (1.5,1.5) -- (1.5,2.5) node[midway, fill=white] {\scriptsize \textvisiblespace};
\draw[gray!50!white] (0.5,1.5) -- (0.5,1.85);
\draw[-stealth] (0.5,2.15) -- (0.5,2.5);
\node at (0.5,2) [fill=none] {\scriptsize \textvisiblespace};
\draw (-0.5,1.5) -- (-0.5,1.82) node[above, fill=none] {\scriptsize $q_0$ \textvisiblespace};
\draw[-stealth] (-0.5, 2.15) -- (-0.5,2.5);

\draw[-stealth, gray!50!white] (-1.5,1.5) -- (-1.5,2.5) node[midway, fill=white] {\scriptsize \textvisiblespace};
\draw[-stealth, gray!50!white] (-2.5,1.5) -- (-2.5,2.5) node[midway, fill=white] {\scriptsize \textvisiblespace};

\draw[-stealth](-0.5, 2.5) -- (-0.5,3.5) node[midway, fill=color4] {\scriptsize $s_i$};
\draw[-stealth](-0.5, 2.5) -- (0.5,2.5) node[midway, fill=color4] {\scriptsize $q_{j}$};
\draw[-stealth](0.5, 2.5) -- (0.5,3.5) node[midway, fill=color4] {\scriptsize $q_{j}$ \textvisiblespace};

\draw[-stealth, gray!50!white] (1.5,2.5) -- (1.5,3.5) node[midway, fill=white] {\scriptsize \textvisiblespace};
\draw[-stealth, gray!50!white] (-1.5,2.5) -- (-1.5,3.5) node[midway, fill=white] {\scriptsize \textvisiblespace};
\draw[-stealth, gray!50!white] (-2.5,2.5) -- (-2.5,3.5) node[midway, fill=white] {\scriptsize \textvisiblespace};

\draw[-stealth](0.5, 3.5) -- (-0.5,3.5) node[midway, fill=color4] {\scriptsize $q_{k}$};
\draw[-stealth](-0.5, 3.5) -- (-0.5,4.5) node[midway, fill=color4] {\scriptsize $q_{k} s_{\ell}$};
\draw[-stealth](0.5, 3.5) -- (0.5,4.5) node[midway, fill=color4] {\scriptsize $q_{r}$};

\draw[-stealth, gray!50!white] (1.5,3.5) -- (1.5,4.5) node[midway, fill=white] {\scriptsize \textvisiblespace};
\draw[-stealth, gray!50!white] (-1.5,3.5) -- (-1.5,4.5) node[midway, fill=white] {\scriptsize \textvisiblespace};
\draw[-stealth, gray!50!white] (-2.5,3.5) -- (-2.5,4.5) node[midway, fill=white] {\scriptsize \textvisiblespace};

\draw[thick,stealth-stealth] (-2.8,-1) -- (1.8,-1) node[below, midway] {Instantaneous description};
\node at (-0.5,-1.6) {of the computation};

\draw[thick,-stealth] (3,1.5) -- (3,4.3) node[below, midway, rotate=90] {Time};

\end{tikzpicture}
\caption{In the reduction $\NHALTAll\to\TILE$, 
the instantaneous description of the Turing machine is mapped to a horizontal configuration of tiles, 
and every computational step is mapped to a valid tiling of the horizontal line above. 
The green tile  is fixed at the origin, while the orange tiles realize the computation. The rest of the plane is filled with trivial tiles, such as the empty tiles (bottom) or tiles copying the tape information (left and right). A Turing machine halts along every path within $n$ steps if and only if the corresponding tiling terminates after $n$ horizontal lines.}
\label{fig:Tiling}
\end{figure}
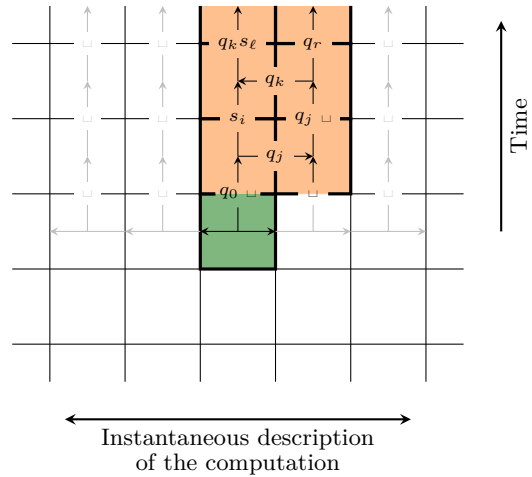

\subsection{Ground State Energy problem}
\label{ssec:GSE}

We now study a version of the ground state energy problem. For this purpose, we consider a spin system on a 2D grid. We assume that every spin takes values in a set $\mathcal{S}$. Given coupling functions $h^{x}, h^{y}: \mathcal{S} \times \mathcal{S} \to \mathbb{N}$ and a local field $h^{\text{loc}}:\mathcal{S} \to \mathbb{N}$, we define the Hamiltonian
$$H_n(\mathbf{s}) = h^{\text{loc}}(s_{00}) + \sum_{\langle \textbf{a},\textbf{b}\rangle_x} h^{x}(s_{\textbf{a}}, s_{\textbf{b}}) + \sum_{\langle \textbf{a},\textbf{b}\rangle_y} h^{y}(s_{\textbf{a}}, s_{\textbf{b}})$$
where $\mathbf{s} = (s_{ij})_{i,j \in \{-n, \ldots, 0, \ldots, n\}}$ is a given spin configuration on the grid $\mathbb{Z}_n^2$ taking values in $\mathcal{S}$ and $s_{\mathbf{a}}, s_{\mathbf{b}}$ denote the elements with coordinates $\mathbf{a}$ and $\mathbf{b}$ in this array. Moreover, 
$\langle \textbf{a}, \textbf{b}\rangle_{x/y}$ denotes all neighbors in $x/y$-direction on $\mathbb{Z}_n^{2}$ where the $\mathbf{a}$ has a smaller $x$/$y$-coordinate than $\mathbf{b}$.
Hence, $H_n$ is translational invariant except for the local field on the spin in the origin.

We start by defining the bounded version of this problem, namely the bounded ground state energy problem $\BGSE$:
\begin{quote}
\emph{Given system size $n \in \mathbb{N}$, non-negative functions $h^{x}, h^{y}$, $h^{\text{loc}}$ and energy $E \in \mathbb{Q}$, 
is the ground state energy $E_{\min}(H_n) > E$?}
\end{quote}
A function $h$ is non-negative if it is non-negative on its whole domain.
Note that $\BGSE$ is indeed a bounded version, as $E_{\min}(H_{n+1}) \geq E_{\min}(H_{n}) > E$ since all couplings are non-negative.
Further note that $\BGSE$ is usually formulated in the negated way, i.e.\ the question is if there exists a spin configuration whose energy is below the threshold $E$.

We now extend $\BGSE$ to an unbounded ground state energy problem $\GSE$:
\begin{quote}
\emph{Given non-negative functions $h^{x}, h^{y}$, $h^{\text{loc}}$ and an energy $E \in \mathbb{Q}$, 
is there an $n \in \mathbb{N}$ such that $E_{\min}(H_n) > E$?}
\end{quote}
Note that $\BGSE$ is the bounded version of $\GSE$ according to \Cref{def:boundedVersion}.

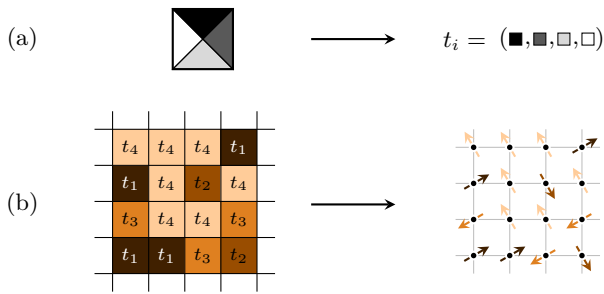
\begin{figure}[thb]\centering
\begin{tikzpicture}[scale=0.8]

\draw[fill=black] (1,1) -- (2,1) -- (1.5,0.5) -- cycle;
\draw[fill=gray!70!black] (2,1) -- (2,0) -- (1.5,0.5) -- cycle;
\draw[fill=gray!30!white] (1,0) -- (2,0) -- (1.5,0.5) -- cycle;

\draw[thick] (1,0) rectangle (2,1);

\draw (1,0) -- (2,1);
\draw (1,1) -- (2,0);

\draw[thick, -stealth] (3.3, 0.5) -- (4.7, 0.5);

\node at (-1.5,0.5) {(a)};

\node at (5.85,0.47) {$t_i = $};
%\node at (0.35,0.47) {$t_i = $};

\node at (6.5,0.5) {$($};
\draw[fill=black] (6.6,0.4) rectangle (6.8,0.6);
\node at (6.9,0.4) {$,$};
\draw[fill=gray!70!black] (7,0.4) rectangle (7.2,0.6);
\node at (7.3,0.4) {$,$};
\draw[fill=gray!30!white] (7.4,0.4) rectangle (7.6,0.6);
\node at (7.7,0.4) {$,$};
\draw[fill=white] (7.8,0.4) rectangle (8,0.6);
\node at (8.1,0.5) {$)$};

%\draw[gray] (0.2,-0.7) -- (7.5,-0.7);

%-------------------------------------------

\draw[thick, -stealth] (3.3, -2.25) -- (4.7, -2.25);
\node at (-1.5,-2.25) {(b)};

\begin{scope}[yshift = -4cm, scale=1.2]

%\draw[draw=none, fill=orange!40!white] (0,0) rectangle (0.5,0.5) node[midway] {\scriptsize $t_1$};
%\draw[draw=none, fill=orange!60!black] (0.5,0) rectangle (1,0.5);
%\draw[draw=none, fill=orange!40!white] (1,0) rectangle (1.5,0.5);
%\draw[draw=none, fill=gray!50!white] (1.5,0) rectangle (2,0.5);
%\draw[draw=none, fill=orange!40!white] (2,0) rectangle (2.5,0.5);

\draw[draw=none, fill=orange!25!black] (0,0.5) rectangle (0.5,1) node[midway] {\scriptsize \color{white}{$t_1$}};
\draw[draw=none, fill=orange!25!black] (0.5,0.5) rectangle (1,1) node[midway] {\scriptsize \color{white}{$t_1$}};
\draw[draw=none, fill=orange!75!gray] (1,0.5) rectangle (1.5,1) node[midway] {\scriptsize {$t_3$}};
\draw[draw=none, fill=orange!60!black] (1.5,0.5) rectangle (2,1) node[midway] {\scriptsize $t_2$};
%\draw[draw=none, fill=orange!75!gray] (2,0.5) rectangle (2.5,1 node[midway] {\scriptsize $t_1$});

\draw[draw=none, fill=orange!75!gray] (0,1) rectangle (0.5,1.5) node[midway] {\scriptsize {$t_3$}};
\draw[draw=none, fill=orange!40!white] (0.5,1) rectangle (1,1.5) node[midway] {\scriptsize $t_4$};
\draw[draw=none, fill=orange!40!white] (1,1) rectangle (1.5,1.5) node[midway] {\scriptsize $t_4$};
\draw[draw=none, fill=orange!75!gray] (1.5,1) rectangle (2,1.5) node[midway] {\scriptsize {$t_3$}};
%\draw[draw=none, fill=orange!25!black] (2,1) rectangle (2.5,1.5);

\draw[draw=none, fill=orange!25!black] (0,1.5) rectangle (0.5,2) node[midway] {\scriptsize \color{white}{$t_1$}};
\draw[draw=none, fill=orange!40!white] (0.5,1.5) rectangle (1,2) node[midway] {\scriptsize $t_4$};
\draw[draw=none, fill=orange!60!black] (1,1.5) rectangle (1.5,2) node[midway] {\scriptsize $t_2$};
\draw[draw=none, fill=orange!40!white] (1.5,1.5) rectangle (2,2) node[midway] {\scriptsize $t_4$};
%\draw[draw=none, fill=orange!40!white] (2,1.5) rectangle (2.5,2);

\draw[draw=none, fill=orange!40!white] (0,2) rectangle (0.5,2.5) node[midway] {\scriptsize $t_4$};
\draw[draw=none, fill=orange!40!white] (0.5,2) rectangle (1,2.5) node[midway] {\scriptsize $t_4$};
\draw[draw=none, fill=orange!40!white] (1,2) rectangle (1.5,2.5) node[midway] {\scriptsize $t_4$};
\draw[draw=none, fill=orange!25!black] (1.5,2) rectangle (2,2.5) node[midway] {\scriptsize \color{white}{$t_1$}};
%\draw[draw=none, fill=orange!60!black] (2,2) rectangle (2.5,2.5);

%\draw (-0.25,0) -- (2.25,0);
\draw (-0.25,0.5) -- (2.25,0.5);
\draw (-0.25,1) -- (2.25,1);
\draw (-0.25,1.5) -- (2.25,1.5);
\draw (-0.25,2) -- (2.25,2);
\draw (-0.25,2.5) -- (2.25,2.5);

\draw (0,0.25) -- (0,2.75);
\draw (0.5,0.25) -- (0.5,2.75);
\draw (1,0.25) -- (1,2.75);
\draw (1.5,0.25) -- (1.5,2.75);
\draw (2,0.25) -- (2,2.75);
%\draw (2.5,-0.25) -- (2.5,2.75);

\begin{scope}[xshift=-1cm]

%\draw[orange!25!black] (5.75,0.25) -- (8.25,0.25);
\draw[gray!50!white] (5.75,0.75) -- (7.75,0.75);
\draw[gray!50!white] (5.75,1.25) -- (7.75,1.25);
\draw[gray!50!white] (5.75,1.75) -- (7.75,1.75);
\draw[gray!50!white] (5.75,2.25) -- (7.75,2.25);

\draw[gray!50!white] (6,0.5) -- (6,2.5);
\draw[gray!50!white] (6.5,0.5) -- (6.5,2.5);
\draw[gray!50!white] (7,0.5) -- (7,2.5);
\draw[gray!50!white] (7.5,0.5) -- (7.5,2.5);
%\draw[gray!50!white] (8,0) -- (8,2.5);

%\draw[thick, orange!40!white, -stealth, rotate around={30:(6,0.25)}] (6,0.1) -- (6,0.5);
%\draw[thick, orange!60!black, -stealth, rotate around={210:(6.5,0.25)}] (6.5,0.1) -- (6.5,0.5);
%\draw[thick, orange!40!white, -stealth, rotate around={30:(7,0.25)}] (7,0.1) -- (7,0.5);
%\draw[thick, gray!50!white, -stealth, rotate around={300:(7.5,0.25)}] (7.5,0.1) -- (7.5,0.5);
%\draw[thick, orange!40!white, -stealth, rotate around={30:(8,0.25)}] (8,0.1) -- (8,0.5);

\draw[thick, orange!25!black, -stealth, rotate around={300:(6,0.75)}] (6,0.6) -- (6,1);
\draw[thick, orange!25!black, -stealth, rotate around={300:(6.5,0.75)}] (6.5,0.6) -- (6.5,1);
\draw[thick, orange!75!gray, -stealth, rotate around={120:(7,0.75)}] (7,0.6) -- (7,1);
\draw[thick, orange!60!black, -stealth, rotate around={210:(7.5,0.75)}] (7.5,0.6) -- (7.5,1);
%\draw[thick, orange!75!gray, -stealth, rotate around={120:(8,0.75)}] (8,0.6) -- (8,1);

\draw[thick, orange!75!gray, -stealth, rotate around={120:(6,1.25)}] (6,1.1) -- (6,1.5);
\draw[thick, orange!40!white, -stealth, rotate around={30:(6.5,1.25)}] (6.5,1.1) -- (6.5,1.5);
\draw[thick, orange!40!white, -stealth, rotate around={30:(7,1.25)}] (7,1.1) -- (7,1.5);
\draw[thick, orange!75!gray, -stealth, rotate around={120:(7.5,1.25)}] (7.5,1.1) -- (7.5,1.5);
%\draw[thick, orange!25!black, -stealth, rotate around={300:(8,1.25)}] (8,1.1) -- (8,1.5);

\draw[thick, orange!25!black, -stealth, rotate around={300:(6,1.75)}] (6,1.6) -- (6,2);
\draw[thick, orange!40!white, -stealth, rotate around={30:(6.5,1.75)}] (6.5,1.6) -- (6.5,2);
\draw[thick, orange!60!black, -stealth, rotate around={210:(7,1.75)}] (7,1.6) -- (7,2);
\draw[thick, orange!40!white, -stealth, rotate around={30:(7.5,1.75)}] (7.5,1.6) -- (7.5,2);
%\draw[thick, orange!40!white, -stealth, rotate around={30:(8,1.75)}] (8,1.6) -- (8,2);

\draw[thick, orange!40!white, -stealth, rotate around={30:(6,2.25)}] (6,2.1) -- (6,2.5);
\draw[thick, orange!40!white, -stealth, rotate around={30:(6.5,2.25)}] (6.5,2.1) -- (6.5,2.5);
\draw[thick, orange!40!white, -stealth, rotate around={30:(7,2.25)}] (7,2.1) -- (7,2.5);
\draw[thick, orange!25!black, -stealth, rotate around={300:(7.5,2.25)}] (7.5,2.1) -- (7.5,2.5);
%\draw[thick, orange!60!black, -stealth, rotate around={210:(8,2.25)}] (8,2.1) -- (8,2.5);

%\draw[fill=black,draw=white, thick] (6, 0.25) circle (1.5pt);
%\draw[fill=black,draw=white, thick] (6.5, 0.25) circle (1.5pt);
%\draw[fill=black,draw=white, thick] (7, 0.25) circle (1.5pt);
%\draw[fill=black,draw=white, thick] (7.5, 0.25) circle (1.5pt);
%\draw[fill=black,draw=white, thick] (8, 0.25) circle (1.5pt);

\draw[fill=black,draw=white, thick] (6, 0.75) circle (1.5pt);
\draw[fill=black,draw=white, thick] (6.5, 0.75) circle (1.5pt);
\draw[fill=black,draw=white, thick] (7, 0.75) circle (1.5pt);
\draw[fill=black,draw=white, thick] (7.5, 0.75) circle (1.5pt);
%\draw[fill=black,draw=white, thick] (8, 0.75) circle (1.5pt);

\draw[fill=black,draw=white, thick] (6, 1.25) circle (1.5pt);
\draw[fill=black,draw=white, thick] (6.5, 1.25) circle (1.5pt);
\draw[fill=black,draw=white, thick] (7, 1.25) circle (1.5pt);
\draw[fill=black,draw=white, thick] (7.5, 1.25) circle (1.5pt);
%\draw[fill=black,draw=white, thick] (8, 1.25) circle (1.5pt);

\draw[fill=black,draw=white, thick] (6, 1.75) circle (1.5pt);
\draw[fill=black,draw=white, thick] (6.5, 1.75) circle (1.5pt);
\draw[fill=black,draw=white, thick] (7, 1.75) circle (1.5pt);
\draw[fill=black,draw=white, thick] (7.5, 1.75) circle (1.5pt);
%\draw[fill=black,draw=white, thick] (8, 1.75) circle (1.5pt);

\draw[fill=black,draw=white, thick] (6, 2.25) circle (1.5pt);
\draw[fill=black,draw=white, thick] (6.5, 2.25) circle (1.5pt);
\draw[fill=black,draw=white, thick] (7, 2.25) circle (1.5pt);
\draw[fill=black,draw=white, thick] (7.5, 2.25) circle (1.5pt);
%\draw[fill=black,draw=white, thick] (8, 2.25) circle (1.5pt);

\end{scope}
\end{scope}

\end{tikzpicture}
\caption{In the reduction $\TILE\to\GSE$, 
(a) every tile $t_i$ is mapped to a  spin state $s_i$. 
(b) Every (valid and invalid) tiling maps to a spin configuration. A tiling of size $n$ is valid iff the corresponding spin configuration is the ground state of $H_n$ with energy $0$.}
\label{fig:GSEReduction}
\end{figure}

Let us show that $\GSE$ is $\RE$-hard and $\BGSE$ is $\coNP$-hard by a reduction $\mathcal{R}:\TILE\to\GSE$ (see \Cref{fig:GSEReduction}).
Given a set of tiles $\mathcal{T} = \{t_1, \ldots, t_k\}$, we define the set of spin states as the set of tiles $\mathcal{S} \coloneqq \mathcal{T}$. Since each tile is specified by four colors in a color space $C$, it can be represented as a $4$-tuple
$$t_{i} = \left(t_{i}^{N}, t_{i}^{E}, t_{i}^{S}, t_{i}^{W}\right)$$
where the entries represent the colors on the top, right, bottom, and left of the tile. We define the coupling function so that a valid tiling with $t_1$ in the origin maps to a spin configuration of energy $0$, 
and every inconsistent color pairing in an invalid tiling gives an additional energy penalty of $1$. More precisely,
$$
h^{x}(s, \hat{s}) \coloneqq 1 - \delta(s^E, \hat{s}^W) \ \text{ and } \ h^{y}(s, \hat{s}) \coloneqq 1 - \delta(s^N, \hat{s}^S).$$
where $s,\hat{s}\in \mathcal{S}$. According the definition of $H_n$, the first component of $h^{x}$ addresses the spin on the left and the second the spin on the right while the first component of $h^{y}$ addresses the spin on the bottom and the second the spin on the top.
Moreover, we define
$$h^{\text{loc}}(s) \coloneqq 1 - \delta(s, t_1).
$$ 
Note that $H_n$ has a ground state of energy zero if and only if there exists a valid tiling of $\mathbb{Z}_n^2$ with tile $t_1$ at the origin. 
That is, $E_{\min}(H_n) > 0$ if and only if there is no valid tiling of size $n$. 
This guarantees that $\mathcal{R}$ is a reduction.  
Additionally, we obtain a reduction from $\BTILE$ to $\BGSE$ since the bounding parameters are identical. 
Similar to the tiling problem, one can show that $\GSE$ is $\RE$-complete and $\BGSE$ is $\coNP$-complete.

Note that non-translational invariant versions of $\BGSE$ are known to be $\coNP$-hard since their negated versions are $\NP$-hard.
In particular,  the ground state energy problem for 2D Ising models with fields is $\NP$-complete \cite{Ba82}.  

\section{Conclusions and Outlook}
\label{sec:questions}

In this work, we have shown a relation between the hardness of an (unbounded) problem and the hardness of its bounded version. In particular, we have defined a bounded version of a language (\Cref{def:boundedVersion}) and given a condition under which a reduction between the unbounded problems translates to a reduction between their bounded versions (\Cref{thm:boundedHardness}). We have also applied this result to two classes of examples (\Cref{sec:examples}): 
First, we showed that $\RE$-hard problems like $\LangPCP$, $\MPO$, or $\REACH$ have an $\NP$-hard bounded version; 
Second, we showed that $\RE$-hard problems like $\TILE$ and $\GSE$ have a $\coNP$-hard bounded version.

It would be interesting to  extend this work to problems in quantum physics such as the spectral gap problem \cite{Cu15,Ba18b} 
or membership problems for quantum correlations \cite{Sl19, Sl19b, Ji21, Fu21, Mo21}. 
A bounded version of the latter uses the dimension of the entangled state as the bounding parameter.

Another open question is whether the undecidability of Diophantine equations \cite{Ma70} and the $\NP$-hardness of its bounded version \cite{Ma76} fits into our framework \footnote{Recall that a Diophantine equation is a polynomial over the integers whose solutions need to be integers.}.
In this context, the unbounded problem is as follows:
\begin{quote}
\emph{Given a Diophantine equation $p(\mathbf{x}, \mathbf{y}) = 0$ with $2k$ variables, 
and a $k$-tuple of integers $\mathbf{a} \in \mathbb{Z}^k$, 
does there exist $\mathbf{b} \in \mathbb{Z}^k$ 
such that $p(\mathbf{a}, \mathbf{b}) = 0$?
}
\end{quote}
Note that here $k$ is fixed. 
The bounded version would restrict to values $\mathbf{b} \in \{-n, \ldots, n\}^k$, where $n$ acts as the bounding parameter. 

Are there also hard bounded versions with other types of complexity, such as $\QMA$-hard \cite{Wa08} bounded versions? While we only considered the scenario of $\RE$-hard problems with either $\NP$-hard or $\coNP$-hard bounded versions, 
there might be ``root problems'' whose bounded version is neither $\NP$-hard or $\coNP$-hard. 
Natural candidates for $\QMA$-hard bounded version are the bounded/unbounded satisfiability problems of quantum circuits \cite{Bo12}, which concerns Turing machines generating polynomial-size quantum circuits. 
The results of this work would imply that certain $\QMA$-hard problems, like the ground state energy problem for $k$-local quantum Hamiltonians \cite{Ke05}, relate to unbounded problems which are undecidable. 

Finally, is it possible to prove the converse direction of \Cref{thm:boundedHardness}? Since bounded languages give rise to a unique unbounded language, can every reduction between bounded versions be transferred to a reduction between the corresponding unbounded problems? If the bounded reduction is of the special form 
\begin{equation*}
\mathcal{R}_b: \langle x, n \rangle \mapsto \langle \mathcal{R}(x), p(n)\rangle
\end{equation*}
with $p$ being a strictly increasing polynomial, then $\mathcal{R}$ is automatically a reduction between the unbounded problems. Yet, the question is open for general $\mathcal{R}_b$.

\section{Acknowledgments}
This project originated at a group retreat together with the group of Hans Briegel and Thomas M\"uller in Obergurgl, and we thank all participants for the fruitful discussions. 
MVDE and AK acknowledge support of the Stand Alone Project P33122-N of the Austrian Science Fund (FWF). 
AK further acknowledges funding of the Austrian Academy of Sciences
(\"OAW) through the DOC scholarship 26547.
SS and TR acknowledge support
of the START Prize Y 1261-N of the Austrian Science Fund (FWF).


%\bibliography{bibliography.bib}

\newpage 
\appendix

\section{Background on computational complexity}
\label{app:comput}

In the following, we summarize the basic notions in computational complexity that are relevant for this paper. For an introduction to the topic, we refer to standard textbooks such as \cite{Ar09, Si06}.

\subsubsection{Deterministic and non-deterministic Turing machines} A \emph{(deterministic) Turing machine} is a model of computation consisting of a head, with an internal state, which operates on an infinitely long tape. In words, it works as follows. 
The input of a Turing machine is initially written on the tape. 
In each computation step, 
the head reads off one entry of the tape, 
it changes its internal state according to the symbol on the tape and its current state, 
it overwrites the symbol on the tape, 
and moves one cell left or right. 
The Turing machine repeats this procedure until it reaches a final state.

More formally, a Turing machine consists of the following: A tape alphabet $\Sigma$ with blank symbol \textvisiblespace \ $\in \Sigma$, a state set $Q$ with an initial state $q_{0}$ and final states $F \subseteq Q$, and a transition function 
$$\delta: (Q \setminus F) \times \Sigma \to Q \times \Sigma \times \{L, R\},
$$
which maps combinations of tape symbol and internal state 
to a new tape symbol and  internal state, 
 together with the instruction to move left or right. 

A \emph{non-deterministic Turing machine} is defined similarly, 
with the only difference that $\delta$ can be a multivalued function, 
i.e.\ a tuple $(x,q)$ can map to multiple state-symbol-direction triples. 
For this reason, a non-deterministic Turing machine has multiple computation paths. 
A non-deterministic Turing machine halts (within $k$ steps) if there is at least one computation path where it halts (within $k$ steps).

In this work, we often consider Turing machines with empty inputs. This means that every entry on the tape is initially given by the blank symbol \textvisiblespace.

\subsubsection{Decision problems and languages} Decision problems are given by a set of instances together with a question that splits the instance set into yes- and no-instances. A language $L \subseteq \Sigma^{*}$ of the problem is defined by encoding the set of instances using an alphabet $\Sigma$ and then collecting all yes-instances as elements in $L$. 
In this work, we address problems and their languages interchangeably via terms like $\NHALT, \LangPCP, \TILE,$ etc.

\subsubsection{Complexity classes and reductions}
Complexity classes are used to characterize the hardness of decision problems. 
 A language $L$ is \emph{decidable} if there exists a Turing machine $T$ which decides membership in finite time, i.e.\ if $x\in L$ then $T$ accepts $x$ in finite time, and 
 if  $x \notin L$ then $T$ rejects it in finite time. 

A language $L$ is \emph{recursively enumerable} (written $L \in \RE$) if and only if there exists a Turing machine $T$ such that for every yes-instance $x \in L$, there exists a finite certificate $y \in \Sigma^{*}$ for verification, i.e.\ $\langle x, y\rangle$ is accepted by $T$ in finite time. 
This means that there exists an algorithm that verifies $x \in L$ in finite time; 
yet, $x \notin L$ may not be rejected in finite time (since the class decidable is different from $\RE$). 

The complexity classes $\P$ and $\NP$ are defined similarly to decidable and $\RE$, respectively, with the only difference that we ask for efficient (i.e.\ in polynomial time) solution or verification. 
Specifically, a problem $L$ is in $\P$ if it can be decided by a polynomial-time deterministic Turing machine, 
i.e.\ a Turing machine that halts on input $x$ within $p(|x|)$ steps, where $p$ is a fixed polynomial and $|x|$ the length of the input string.
A problem $L$ is in $\NP$ 
if and only if 
there is a polynomial-time Turing machine $T$ such that 
\begin{equation}
\label{eq:NP}
x \in L \quad \Longleftrightarrow \quad \exists y \in \Sigma^{p(|x|)}: T \text{ accepts }\langle x, y\rangle.
\end{equation}
A language $L$ is in $\coNP$ if and only if its complement $L^{c} = \Sigma^*\setminus L$ is in $\NP$.
Specifically, $L$ is in $\coNP$ if and only if there exists a polynomial-time Turing machine $T'$ such that
\begin{equation}
\label{eq:coNP}
x \in L \quad \Longleftrightarrow \quad \forall y \in \Sigma^{p(|x|)}: T' \text{ accepts }\langle x, y\rangle,
\end{equation}
Given a polynomial-time Turing machine $T$ verifying $L^{c}$, $T'$ accepts if and only if $T$ rejects, which proves the statement.

In this work, we are mainly interested in hardness results. 
For a given complexity class $\mathcal{C}$, a problem is $\mathcal{C}$-hard if it is (in a formal way) at least as hard as any problem in $\mathcal{C}$. 
The problem is $\mathcal{C}$-complete if it is $\mathcal{C}$-hard and in $\mathcal{C}$. 

To formalize these concepts, we need the notion of a reduction between languages.  
A \emph{reduction from $L'$ to $L$} is a Turing-computable function $\mathcal{R}: \Sigma^{*} \to \Sigma^{*}$ which satisfies
$$ x \in L' \quad \Longleftrightarrow \quad \mathcal{R}(x) \in L.$$
By abuse of notation, we write  $\mathcal{R}: L' \to L$ to highlight the source and target language. 
If there exists a polynomial-time reduction from $L'$ to $L$, we  write $L' \leq_{\poly} L$.
Now we can define $\RE$-hardness: 
A problem $L$ is $\RE$-hard if there exists a reduction $L' \to L$ for every problem $L' \in \RE$. 
$L$ is $\RE$-complete if  $L$ is $\RE$-hard and  $L \in \RE$.
$\RE$-complete problems are in a formal sense the hardest problems in $\RE$. 
Similarly, a problem $L$ is $\NP$-hard if for every problem $L' \in \NP$ there exists a polynomial-time reduction $\mathcal{R}: L' \to L$. 
$L$ is $\NP$-complete if it is $\NP$-hard and $L \in \NP$. 
One defines $\coNP$-hardness analogously.

\section{Complexity of (bounded) halting problems}
\label{app:HaltingProofs}

In this appendix, we provide a detailed analysis of the two halting problems $\NHALT$ and $\NHALTAll$ together with their bounded versions which act as root problems in the main text (see \Cref{sec:root}). We start with the unbounded problems and their undecidability, and continue with their bounded versions and their complexity.

We start by noting that the input of  $\NHALT$ and $\NHALTAll$ is just a Turing machine $T$, as we ask whether $T$ halts on the empty tape.

\begin{definition}
Let $T$ be a description of a non-deterministic Turing machine. 
\begin{align*}
T \in \NHALT &\quad :\Longleftrightarrow \quad T \text{ halts on the empty tape.}\\
T \in \NHALTAll &\quad :\Longleftrightarrow \quad 
\begin{array}{l} T \text{ halts on the empty tape} \\ \text{along all paths.}\end{array}
\end{align*}
\end{definition}

Both problems are undecidable, as the following reduction from the halting problem $\HALT$ shows. 

\begin{theorem}
\label{thm:bnhalt}
$\NHALT$ and $\NHALTAll$ are $\RE$-complete.
\end{theorem}
\begin{proof}
We prove $\RE$-hardness only for $\NHALT$, as the same argument applies to $\NHALTAll$.
To this end, we provide a reduction from $\HALT$.  
Recall that $\HALT$ takes $\langle T,x_0\rangle$ as input  (where $T$ is a description of a deterministic Turing machine $T$ \footnote{Note that a deterministic Turing machine is a special case of a non-deterministic Turing machine only having one computational path}, and $x_0$ is an input) 
and accepts if and only if $T$ halts on $x_0$. 
The reduction transforms instance $\langle T,x_0\rangle$ 
to a  Turing machine $T'= \mathcal{R}(\langle T,x_0\rangle) $ which first writes $x_0$ on the tape, 
and then does the same computation as $T$ on the given input.  
By construction, $\langle T,x\rangle \in \HALT$ if and only if $T' \in \NHALT$, i.e.\ $\mathcal{R}$ is a valid reduction.

That $\NHALT\in \RE$ follows by taking the halting computation path as a certificate, 
and a verifier that verifies the computation along the path. 
That $\NHALTAll\in \RE$ follows by taking the halting time as a certificate, 
and a verifier that verifies that the computation halts along all paths within this halting time.
\end{proof}

Let us now consider the bounded versions $\BNHALT$ and $\BNHALTAll$. 
Since these problems have different complexity we will treat them separately.

\begin{definition}
Let $T$ be a description of a non-deterministic Turing machine, and $n \in \mathbb{N}$. 
$$
\langle T, 1^n \rangle \in \BNHALT \ \ :\Longleftrightarrow \ \ \begin{array}{l} T 
\text{ halts on the empty tape} \\ \text{in } n \text{ steps.} \end{array}
$$
\end{definition}

\begin{theorem}
\label{thm:bnhalt}
$\BNHALT$ is $\NP$-complete.
\end{theorem}
\begin{proof}
To show that $\BNHALT$ is $\NP$-hard, we prove that every $\NP$-language $L$ has a polynomial-time reduction to $\BNHALT$. 
Since $L$ is in $\NP$, 
 there exists a non-deterministic polynomial-time Turing machine $M$ which accepts $x$ within time $p(|x|)$ if and only if $x\in L$.  
We construct a non-deterministic Turing machine $P_{M,x}$ that 
(i) writes $x$ on the tape, 
(ii) does the same computation as $M$ on the tape with input $x$, 
and (iii) 
if $M$ accepts  $x$ along a path, $P_{M,x}$ halts along this path, 
and if $M$ rejects $x$ along a path, $P_{M,x}$ loops along this path. 
Since step (i) needs a polynomial number $q(|x|)$ steps, 
 and step (iii) needs a constant number $k$ of steps, 
we have that $x \in L$ if and only if $\langle P_{M,x}, 1^{q(|x|) + k + p(|x|)}\rangle \in \BNHALT$.
Completeness follows from Equation \eqref{eq:NP} by choosing the halting computation path as a certificate,
 and a polynomial-time verifier which verifies the computation along this path. 
\end{proof}

Similarly, we define the problem $\BNHALTAll$ as the language accepting the instance $\langle T, 1^n \rangle$ if and only if $T$ halts on the empty tape along \emph{all} computation paths in at most $n$ steps.

\begin{definition}
Let $T$ be a description of a non-deterministic Turing machine $T$, and $n \in \mathbb{N}$.
$$
\langle T, 1^n \rangle \in \BNHALTAll  \quad :\Longleftrightarrow \quad \begin{array}{l} T \text{ halts on} \\ \text{the empty tape} \\ \text{along all paths} \\ \text{in } n \text{ steps.}\end{array}
$$
\end{definition}

\begin{theorem}
\label{thm:bnhalt-all}
$\BNHALTAll$ is $\coNP$-complete.
\end{theorem}

\begin{proof}
The hardness proof is very similar to \Cref{thm:bnhalt}. 
Namely, we prove that every $\coNP$-language $L$ has a polynomial-time reduction to $\BNHALTAll$. 
Since $L$ is in $\coNP$, 
there exists a non-deterministic polynomial-time Turing machine $M$ which accepts $x$ along every computation path of length at most $p(|x|)$ if and only if $x\in L$.  
We construct the non-deterministic Turing machine $P_{M,x}$ which 
(i) writes $x$ on the tape,
(ii) does the same computation as $M$ on the tape with input $x$,
and 
(iii) if $M$ accepts $x$ along a path, $P_{M,x}$ halts along this path. If $M$ rejects $x$ along a path, $P_{M,x}$ loops along this path. 
Since (i) needs a polynomial number $q(|x|)$ steps and (iii) needs a constant number $k$ of steps, 
we have that $x \in L$ if and only if $\langle P_{M,x}, 1^{q(|x|) + k + p(|x|)}\rangle \in \BNHALTAll$.
Completeness again follows from Equation \eqref{eq:coNP} by choosing computation paths as a certificate,
 and a polynomial-time verifier that verifies the computation along the given path.
\end{proof}

\section{More details on undecidable problems and their bounded versions }
\label{app:Examples}

In this appendix we provide more details on the undecidable problems and their bounded versions considered in the main text. 
Specifically, we consider the PCP problem (\cref{app:PCP}),
the zero in the upper left corner (\cref{app:ZULC}), 
the matrix mortality problem (\cref{app:MM})
the MPO positivity problem (\cref{app:MPO}),
 the polynomial positivity problem (\cref{app:POLY}),
the stability of positive maps  (\cref{app:STAB}), 
and   the tiling problem (\cref{app:Tiling}).

\subsection{The PCP problem}
\label{app:PCP}

We now provide the reduction  $\NHALT\to \LangPCP$ in greater detail. 
The following reduction modifies that of Ref.\ \cite{Si06}, so that the bounding parameters of both problems are  polynomially related.

We consider a Turing machine given by a tape alphabet $\Sigma$ with blank symbol \textvisiblespace \ $\in \Sigma$, a state set $Q$ with an initial state $q_{0}$, final states $F \subseteq Q$, and a transition function 
$$\delta: \Sigma \times (Q \setminus F) \to \Sigma \times Q \times \{L, R\}.$$
Without loss of generality, we consider here only semi-infinite tape Turing machines, i.e.\ having a tape with a left end but no right end. This is no restriction for the complexity since semi-infinite tape Turing machines are equivalent to standard Turing machines \cite[Claim 1.4]{Ar09}.

This Turing machine is mapped to the following set of dominoes $\mathcal{D}$:

\medskip\noindent
\renewcommand{\arraystretch}{2.5}
\setlength{\tabcolsep}{0.2em}
\begin{tabular}{r l c} 
(i) & An initial domino & \adjustbox{valign=m}{\begin{tikzpicture} \draw[thick, shade, blur shadow={shadow blur steps=8,shadow blur extra rounding=1.3pt, shadow scale=0.96}] (-2,-0.5) -- (0,-0.5) -- (0,-1.8) -- (-2,-1.8) -- cycle;
\draw[fill=white] (-2,-0.5) -- (0,-0.5) -- (0,-1.15) -- (-2,-1.15) -- cycle;
\draw[fill=white] (-2,-1.15) -- (0,-1.15) -- (0,-1.8) -- (-2,-1.8) -- cycle;
\draw[thick] (-2,-0.5) -- (0,-0.5) -- (0,-1.8) -- (-2,-1.8) -- cycle;
\draw (-2, -1.15) -- (0, -1.15);

\node at (-1, -0.82) {\scriptsize $!$};
\node at (-1, -1.52) {\scriptsize $! \, \gstar \, q_0 \, \gstar $ \textvisiblespace \ $ \gstar \, ! \, \gstar$};
\end{tikzpicture}} \\[0.5cm]
(ii) & For every $x \in \Sigma$, a copy domino  & \adjustbox{valign=m}{\begin{tikzpicture} \draw[draw=none, shade, blur shadow={shadow blur steps=8,shadow blur extra rounding=1.3pt, shadow scale=0.96}] (0.75,-0.5) -- (1.5, -0.5) -- (1.5,-1.8) -- (0.75,-1.8) -- cycle;
\draw[draw=none, fill=white] (0.75,-0.5) -- (1.5, -0.5) -- (1.5,-1.15) -- (0.75,-1.15) -- cycle;
\draw[draw=none, fill=white] (0.75,-1.8) -- (1.5, -1.8) -- (1.5,-1.15) -- (0.75,-1.15) -- cycle;
\draw[thick] (0.75,-0.5) -- (1.5, -0.5) -- (1.5,-1.8) -- (0.75,-1.8) -- cycle;

\draw (0.75, -1.15) -- (1.5, -1.15);

\node at (1.125,-0.82) {\scriptsize $\gstar \, x$};
\node at (1.125,-1.52) {\scriptsize $x \, \gstar $};
\end{tikzpicture}}\\[0.5cm]

(iii) & Transitions $(q,x) \mapsto (\hat{q},y,L)$ &  \adjustbox{valign=m}{\begin{tikzpicture} \draw[draw=none,fill=white, shade, blur shadow={shadow blur steps=8,shadow blur extra rounding=1.3pt, shadow scale=0.96}] (1.05,-0.5) -- (2.4,-0.5) -- (2.4,-1.8) -- (1.05,-1.8) -- cycle;
\draw[draw=none, fill=white] (1.05,-0.5) -- (2.4,-0.5) -- (2.4,-1.15) -- (1.05,-1.15) -- cycle;
\draw[draw=none, fill=white] (1.05,-1.8) -- (2.4,-1.8) -- (2.4,-1.15) -- (1.05,-1.15) -- cycle;
\draw[thick] (1.05,-0.5) -- (2.4,-0.5) -- (2.4,-1.8) -- (1.05,-1.8) -- cycle;
\draw (1.05, -1.15) -- (2.4, -1.15);

\node at (1.7, -0.82) {\scriptsize $\gstar \, x \, \gstar \, q$};
\node at (1.7, -1.52) {\scriptsize $\hat{q} \, \gstar \, y \, \gstar$};
\end{tikzpicture}}\\[0.5cm]

(iv) &  Transitions $(q,x) \mapsto (\hat{q},y,R)$ &  \adjustbox{valign=m}{\begin{tikzpicture} \draw[draw=none,fill=white, shade, blur shadow={shadow blur steps=8,shadow blur extra rounding=1.3pt, shadow scale=0.96}] (1.05,-0.5) -- (2.4,-0.5) -- (2.4,-1.8) -- (1.05,-1.8) -- cycle;
\draw[draw=none, fill=white] (1.05,-0.5) -- (2.4,-0.5) -- (2.4,-1.15) -- (1.05,-1.15) -- cycle;
\draw[draw=none, fill=white] (1.05,-1.8) -- (2.4,-1.8) -- (2.4,-1.15) -- (1.05,-1.15) -- cycle;
\draw[thick] (1.05,-0.5) -- (2.4,-0.5) -- (2.4,-1.8) -- (1.05,-1.8) -- cycle;
\draw (1.05, -1.15) -- (2.4, -1.15);

\node at (1.7, -0.82) {\scriptsize $\gstar \, q \, \gstar \, x$};
\node at (1.7, -1.52) {\scriptsize $y \, \gstar \, \hat{q} \, \gstar$};
\end{tikzpicture}}\\[0.5cm]

(v) & A tape expander & \adjustbox{valign=m}{\begin{tikzpicture} 
\draw [draw=none, shade, blur shadow={shadow blur steps=8,shadow blur extra rounding=1.3pt, shadow scale=0.96}] (0,-0.5) -- (1.2,-0.5) -- (1.2, -1.8) -- (0, -1.8) -- cycle;
\draw [draw=none, fill=white] (0,-0.5) -- (1.2,-0.5) -- (1.2, -1.15) -- (0, -1.15) -- cycle;
\draw [draw=none, fill=white] (0,-1.8) -- (1.2,-1.8) -- (1.2, -1.15) -- (0, -1.15) -- cycle;
\draw [thick] (0,-0.5) -- (1.2,-0.5) -- (1.2, -1.8) -- (0, -1.8) -- cycle;
\draw (0, -1.15) -- (1.2, -1.15);

\node at (0.6,-0.82) {\scriptsize $\gstar \, !$};
\node at (0.6,-1.52) {\scriptsize \textvisiblespace \, $\gstar \, ! \, \gstar$};
\end{tikzpicture}}\\[0.5cm]

(vi) & For every $q_f \in F$, $y_1,y_2 \in \Sigma$ & \adjustbox{valign=m}{\begin{tikzpicture}
\draw [draw=none, shade, blur shadow={shadow blur steps=8,shadow blur extra rounding=1.3pt, shadow scale=0.96}] (2,1) -- (4,1) -- (4,-0.3) -- (2,-0.3) -- cycle;
\draw [draw=none, fill=white] (2,1) -- (4,1) -- (4,0.35) -- (2,0.35) -- cycle;
\draw [draw=none, fill=white] (2,-0.3) -- (4,-0.3) -- (4,0.35) -- (2,0.35) -- cycle;
\draw [thick] (2,1) -- (4,1) -- (4,-0.3) -- (2,-0.3) -- cycle;
\draw (2, 0.35) -- (4, 0.35);

\node at (3, 0.65) {\scriptsize $\gstar \, y_{\scaleto{1}{3pt}} \, \gstar \, q_{\scaleto{f}{3pt}} \, \gstar \, y_{\scaleto{2}{3pt}}$};
\node at (3, 0) {\scriptsize $q_{\scaleto{f}{3pt}} \, \gstar$};
\end{tikzpicture}}\\[0.5cm]

(vii) & For every $q_f \in F$, $y_1,y_2 \in \Sigma$ & \adjustbox{valign=m}{\begin{tikzpicture}
\draw [draw=none, shade, blur shadow={shadow blur steps=8,shadow blur extra rounding=1.3pt, shadow scale=0.96}] (2,1) -- (4,1) -- (4,-0.3) -- (2,-0.3) -- cycle;
\draw [draw=none, fill=white] (2,1) -- (4,1) -- (4,0.35) -- (2,0.35) -- cycle;
\draw [draw=none, fill=white] (2,-0.3) -- (4,-0.3) -- (4,0.35) -- (2,0.35) -- cycle;
\draw [thick] (2,1) -- (4,1) -- (4,-0.3) -- (2,-0.3) -- cycle;
\draw (2, 0.35) -- (4, 0.35);

\node at (3, 0.65) {\scriptsize $\gstar \, q_{\scaleto{f}{3pt}} \, \gstar \, y_{\scaleto{1}{3pt}} \, \gstar \, y_{\scaleto{2}{3pt}}$};
\node at (3, 0) {\scriptsize $q_{\scaleto{f}{3pt}} \, \gstar$};
\end{tikzpicture}}\\[0.5cm]

(viii) & A final domino & \adjustbox{valign=m}{\begin{tikzpicture}
\draw [draw=none, shade, blur shadow={shadow blur steps=8,shadow blur extra rounding=1.3pt, shadow scale=0.96}] (2,1) -- (4,1) -- (4,-0.3) -- (2,-0.3) -- cycle;
\draw [draw=none, fill=white] (2,1) -- (4,1) -- (4,0.35) -- (2,0.35) -- cycle;
\draw [draw=none, fill=white] (2,-0.3) -- (4,-0.3) -- (4,0.35) -- (2,0.35) -- cycle;
\draw [thick] (2,1) -- (4,1) -- (4,-0.3) -- (2,-0.3) -- cycle;
\draw (2, 0.35) -- (4, 0.35);

\node at (3, 0.65) {\scriptsize $ \gstar \, q_{\scaleto{f}{3pt}} \, \gstar \,$ \textvisiblespace \ $\, \gstar \, ! \, \gstar \, !$};
\node at (3, 0) {\scriptsize $!$};
\end{tikzpicture}}
\end{tabular}
\renewcommand{\arraystretch}{1.2}

\medskip

\noindent Note that the domino set $\mathcal{D}$ can be constructed in polynomial time from $T$, and that $|\mathcal{D}|$ is polynomial in $|Q|$ and $|\Sigma|$.

Let us now apply this reduction to a non-deterministic Turing machine, 
as the bounded version needs the latter. 
First note that the exclamation marks serve as a separator between the instantaneous descriptions of different computation steps, 
while the grey star separates every symbol in the string. 
The lower part of the initial domino (i)   represents  the initial tape configuration of the Turing machine 
together with its current head state and position. 
Since the initial domino (i) is the only domino whose first upper and lower symbols coincide, every match has to start with the initial domino. 
A computation step along some computation path is simulated by applying copy-dominoes (ii), transition dominoes (iii), (iv), and tape expanders (v), according to \Cref{fig:PCPReduction}. If a computation reaches a final state $q_f$, 
the final instantaneous description is  successively removed  by applying dominoes (ii), (vi), (vii), and (v) according to \Cref{fig:PCPReduction}. 
Finally, a match is obtained by adding  (viii).

This implies that $T$ halts on the empty tape along a computation path if and only if $\mathcal{D}$ forms a match. 
Hence, $\mathcal{R} \colon \NHALT\to \LangPCP$ is a reduction. It follows that $\LangPCP$ is $\RE$-hard.

Note that simulating the $k^{\text{th}}$ computation step by a domino arrangement requires precisely $k+1$ dominoes. When $T$ reaches the final state after $n$ computation steps, the post-simulation procedure requires another $n+1$ repetitions,  where each procedure needs precisely $m = n + 1$ arrangements with length starting with $m$ and decreasing by $1$. So $T$ halts after $n$ computation steps on the empty tape if and only if the corresponding domino set forms a match in at most
\begin{align*}
q(n) \coloneqq 1 + \sum_{k=1}^{n} (k+1) + \sum_{k=1}^{n + 1} k = (n+1) \cdot (n+2)
\end{align*}
steps, where the first sum represents the computation procedure and the second sum the post-simulation procedure. Since $\mathcal{R}$ is a polynomial-time reduction, using \Cref{thm:boundedHardness}, this implies that
$$
\langle T, 1^n \rangle \mapsto \langle \mathcal{R}(T), 1^{(n+1) \cdot (n+2)} \rangle
$$
is a polynomial-time reduction from $\BNHALT$ to $\BPCP$, which shows that $\BPCP$ is $\NP$-hard.

\subsection{The Zero in the upper left corner problem} 
\label{app:ZULC}

We now present the reduction $\mathcal{R}: \LangPCP \to \ZULC$ based on the ideas of \cite{Ha01b}. 
For this purpose, we consider $\LangPCP$ using strings encoded in the alphabet $\Sigma = \{0,1,2\}$. We define the bijection $\sigma: \Sigma^{*} \to \mathbb{N}$ that assigns a representation in base $3$ to every natural number, i.e.\
$$ 
\sigma(c_1, \ldots, c_n) \coloneqq \sum_{i=1}^{n} c_i \cdot 3^{n-i}.
$$
Moreover, we define a function $\gamma: \Sigma^{*} \times \Sigma^{*} \to \mathbb{N}^{3 \times 3}$ via
$$
\gamma(w_1, w_2) \coloneqq \begin{pmatrix} 3^{|w_1|} & 0 & 0 \\ 0 & 3^{|w_2|} & 0 \\ \sigma(w_1) & \sigma(w_2) & 1 \end{pmatrix}.
$$

The function $\gamma$ is injective and a morphism, i.e.\ $\gamma(w_1 u_1, w_2 u_2) = \gamma(w_1, w_2) \cdot \gamma(u_1, u_2)$ where composition on $\Sigma^{*}$ 
is given by concatenation of words.
Let
$$d_1 = \left[ \frac{a_1}{b_1} \right], \ldots, d_k = \left[ \frac{a_k}{b_k} \right]$$
be an instance of $\LangPCP$ where $a_i, b_i \in \Sigma^{*}$. For $i \in \{1, \ldots, k\}$, we define the matrices
$$A_i = X \cdot \gamma(a_i, b_i) \cdot X^{-1} \quad B_i = X \cdot \gamma(a_i, 0 b_i) \cdot X^{-1}$$
with
$$
X = \begin{pmatrix} 1 & 0 & 1 \\ 1 & 1 & 0 \\ 0 & 0 & 1 \end{pmatrix}.
$$
We have that
$$d_{i_1} d_{i_2} \cdots d_{i_n}$$
is a matching domino if and only if
$$
(M_{i_1} \cdot M_{i_2} \cdots M_{i_n})_{11} = 0
$$
where $M_{i_j } \in \{ A_{i_j },   B_{i_j }\}$. 
We refer to \cite{Ha01b} for details. This shows that $\mathcal{R}: \LangPCP \to \ZULC$ 
with 
$$ 
\mathcal{R}\Big(\langle d_1, \ldots, d_{k} \rangle \Big) \coloneqq \langle A_1, \ldots, A_k, B_1, \ldots, B_k\rangle  
$$
is  a polynomial-time reduction. 
 This implies that $\ZULC$ is $\RE$-hard.

Since matches of length $n$ are mapped to matrix multiplications of length $n$ with a zero in the upper left corner, this shows that $\mathcal{R}_b: \BPCP \to \BZULC$ with
$$ \mathcal{R}_{b}\Big(\langle d_1, \ldots, d_{k}, 1^n\rangle\Big) \coloneqq \langle A_1, \ldots, A_k, B_1, \ldots, B_k, 1^n\rangle$$
is a polynomial-time reduction. 
This implies that $\BZULC$ is $\NP$-hard.

Note that the  matrices in $A_1, \ldots, A_k, B_1, \ldots, B_k$ are invertible, from which it follows that $\ZULC$ and $\BZULC$ remain $\RE$-hard and $\NP$-hard, respectively, when restricting the instances to invertible matrices.

\subsection{The Matrix mortality problem}
\label{app:MM}

We now construct the reduction $\mathcal{Q}: \ZULC \to \MM$ following the ideas of  \cite{Ha01b}. 
Since $\ZULC$ remains hard when restricting the instances to invertible matrices, we construct $\mathcal{Q}$ only for invertible matrices.
So let $\langle A_1, \ldots, A_k \rangle$ be an instance of invertible matrices in $\ZULC$. We define
$$\mathcal{Q}(\langle A_1, \ldots, A_k \rangle) \coloneqq \langle A_1, \ldots, A_k, B \rangle$$
with 
$$B = \begin{pmatrix} 1 & 0 & 0 \\ 0 & 0 & 0 \\ 0 & 0 & 0 \end{pmatrix}. 
$$

We claim that $A_1, \ldots, A_k$ forms a zero in the upper left corner if and only if $A_1, \ldots, A_k, B$ multiplies to a zero matrix. This proves that $\MM$ is $\RE$-hard. Moreover, we show that
\begin{equation}
\label{eq:threshold-map}
n_{\min, \MM}[\langle \mathbf{A}, B\rangle] = n_{\min, \ZULC}[\langle \mathbf{A} \rangle] + 2.
\end{equation}
where $\mathbf{A}$ represents the list $A_1, \ldots, A_k$. 

To prove the claim, first note that if 
$$
(A_{i_1} \cdot A_{i_2} \cdots A_{i_n})_{11} = 0,
$$
then 
$$
B \cdot A_{i_1} \cdot A_{i_2} \cdots A_{i_n} \cdot B = (A_{i_1} \cdot A_{i_2} \cdots A_{i_n})_{11} = 0.
$$
In other words, a yes-instance of $\ZULC$ with parameter $n$ is mapped to a yes-instance in $\MM$ with parameter $n+2$. This proves the inequality ``$\leq$'' of Equation \eqref{eq:threshold-map}. 

Conversely, assume that there exists a sequence of $n$ matrices in $\{A_1, \ldots A_k, B\}$ that multiplies to $\mathbf{0}$. Since $A_1, \ldots, A_k$ are invertible and $B$ has rank $1$, this sequence must contain $B$ at least twice. The product is of the form 
$$
M_{1} B M_{2} B M_{3} B \cdots B M_{r} = \mathbf{0}
$$
where $M_{i}$ is a multiplication of $\ell_{i}$ matrices in $\{A_1, \ldots, A_k\}$ for some $\ell_{i}$ \footnote{If it is an empty multiplication (i.e.\ $\ell_i = 0$), then we define $M_{i}$ as the identity matrix.}.
Since $B$ is idempotent, we have that
\begin{align*}
0 &= 
\big(M_{1} B M_{2} B M_{3} B \cdots B M_{r}\big)_{11} \\ 
&=  \big(BM_{1} B^2 M_{2} B^2 M_{3} B^2 \cdots B^2 M_{r}B\big)_{11} \\ 
&= \big(M_{1}\big)_{11} \cdots \big(M_{r}\big)_{11}.
\end{align*}
This implies that at least one of the matrices $M_{i}$ has a zero in the upper left corner, which shows that $A_1, \ldots, A_k$ form a zero in the upper left corner with a word of length $n$. 
Specifically, any minimal sequence of matrices realizing $\mathbf{0}$ must be of the form
$$
B \cdot A_{i_1} \cdot A_{i_2} \cdots A_{i_n} \cdot B = \mathbf{0} .
$$
Note that a shorter such product cannot exist because it would violate the proven inequality ``$\leq$'' of Equation \eqref{eq:threshold-map}. 
This representation proves the inequality ``$\geq$'' of Equation \eqref{eq:threshold-map}, since
$$(A_{i_1} \cdot A_{i_2} \cdots A_{i_n})_{11} = 0.$$

In summary, 
$\mathcal{Q}\colon \ZULC\to \MM$ is a reduction, which proves that $\MM$ is $\RE$-hard. Moreover, $\mathcal{Q}_{b}\colon  \BZULC \to \BMM$ with 
$$\mathcal{Q}_{b}: \langle A_1, \ldots, A_k, 1^n \rangle \mapsto \langle A_{1}, \ldots, A_k, B, 1^{n+2}\rangle$$
is a polynomial-time reduction too, which proves that $\BMM$ is $\NP$-hard.

\subsection{The MPO positivity problem}
\label{app:MPO}

Here we present a reduction $\mathcal{R}: \ZULC \to \MPO$, slightly different than \cite{De15}. The $\MPO$ problem has as input a fixed number of $D \times D$ integer matrices $\langle B_i : i \in \{1, \ldots, k\} \rangle$ and asks whether there exists a natural number $n \in \mathbb{N}$ such that
$$
\rho_n(B) \coloneqq \sum_{i_1, \ldots, i_n = 1}^{k} \tr\left(B_{i_1} \cdots B_{i_n}\right) \ket{i_1\ldots i_n} \bra{i_1 \ldots i_n}
$$
is not positive semidefinite. We define
$$\mathcal{R}(\langle A_1, \ldots, A_k\rangle) = \langle B_1, \ldots, B_k, B_{k+1}\rangle$$
where for $i\in \{1,\ldots, k\}$ 
$$B_i \coloneqq \begin{pmatrix} A_i \otimes A_i & 0 \\ 0 & 1\end{pmatrix}$$
and
$$
B_{k+1} \coloneqq \begin{pmatrix} E_{11} & 0 \\ 0 & -1 \end{pmatrix}
$$
where 
$E_{11} \coloneqq \mathbf{e}_1 \mathbf{e}_1^t $
with $ \mathbf{e}_1 =(1,0,\ldots,0)^t$ of length $D$. 

We now prove that  the threshold parameter $n$ in $\BZULC$ maps to the threshold parameter $n+1$ in $\BMPO$.
Let $A_{i_1}, \ldots, A_{i_n}$ be the minimal sequence such that
$$ \big(A_{i_1} \cdot A_{i_2} \cdots A_{i_n}\big)_{11} = 0.$$
Then,
$$\tr(B_{i_1} \cdots B_{i_n} \cdot B_{k+1}) = \big(A_{i_1} \cdot A_{i_2} \cdots A_{i_n}\big)_{11}^2 - 1 < 0.$$
Conversely, let $B_{i_1}, \ldots, B_{i_{n+1}}$ be a minimal sequence such that
$$
\tr(B_{i_1} \cdot B_{i_2} \cdots B_{i_{n+1}}) < 0. 
$$
The indices $i_1, \ldots, i_{n+1}$ cannot be chosen exclusively from $\{1, \ldots, k\}$, since in that case 
$$
\tr(B_{i_1} \cdot B_{i_2} \cdots B_{i_{n+1}}) = \big(\tr(A_{i_1} \cdots A_{i_{n+1}})\big)^2 + 1 \geq 0.
$$
Hence, there is at least one index $i_{\ell} = k+1$. Assume that there is precisely one index $k+1$.
Without loss of generality, we assume $i_{n+1} = k+1$ due to cyclicity of the trace.
This leads to
$$0 > \tr(B_{i_1} \cdot B_{i_2} \cdots B_{i_{n+1}}) = \big(\left(A_{i_1} \cdot A_{i_2} \cdots A_{i_n}\right)_{11} \big)^2 - 1$$
which implies that $\left(A_{i_1} \cdot A_{i_2} \cdots A_{i_n}\right)_{11} = 0$ because the entries are integer. This shows that a threshold parameter $n+1$ in $\BMPO$ maps to a threshold parameter of a most $n$ in $\BZULC$. Note that having multiple indices with $k+1$ leads to a smaller threshold parameter in $\BZULC$ which contradicts the minimality assumption of $B_{i_1}, \ldots, B_{i_{n+1}}$. This proves the statement.

This reduction can easily be extended to matrices with rational numbers. 

In summary, 
$\mathcal{R}\colon \ZULC\to \MPO$ is a reduction, which proves that $\MPO$ is $\RE$-hard. Moreover, by \Cref{thm:boundedHardness}, $\mathcal{R}_{b}\colon  \BZULC \to \BMPO$ with 
$$\mathcal{R}_{b}: \langle A_1, \ldots, A_k, 1^n \rangle \mapsto \langle B_{1}, \ldots, B_k, B_{k+1}, 1^{n+1}\rangle$$
is a polynomial-time reduction too, which proves that $\BZULC$ is $\NP$-hard.

\subsection{The Polynomial positivity problem}
\label{app:POLY}

Let us now review the reduction $\mathcal{R}: \MPO \to \POLY$ from \cite{De21b}. We define 
$$\mathcal{R}(\langle B_1, \ldots, B_k\rangle) \coloneqq \left\langle q_{\alpha, \beta}(\mathbf{x}) : \alpha, \beta = 1, \ldots, D \right\rangle$$
with
$$
q_{\alpha, \beta}(\underline{x}) \coloneqq \sum_{j=1}^{k} \left(B_j\right)_{\alpha, \beta} x_j^2,
$$
where $\underline{x}$ is a $k$-tuple of variables. It is clear that $\mathcal{R}$ is a polynomial-time function. We now prove that $\mathcal{R}$ is a reduction. 

If there exists a sequence of matrices such that
$$\tr(B_{i_1} \cdot B_{i_2} \cdots B_{i_n}) < 0$$
then
$$
p_n(\mathbf{e}_{i_1}, \ldots, \mathbf{e}_{i_n}) = \tr(B_{i_1} \cdot B_{i_2} \cdots B_{i_n}) < 0
$$
where $p_n$ is defined in \eqref{eq:pn}  
and $\mathbf{e}_\ell$ is the $\ell$th standard vector. This implies that $p_n$ is \emph{not} a nonnegative function. Conversely, if
$$\tr(B_{i_1} \cdot B_{i_2} \cdots B_{i_n}) \geq 0$$
for all indices $i_1, \ldots, i_n$, then $p_n$ is a sum-of squares which is also nonnegative.

This proves that the threshold $n$ for $\BMPO$ is mapped to the threshold $n$ for $\BPOLY$. 
It follows  that $\BPOLY$ is $\NP$-hard. 
Moreover, 
$\POLY$ is $\RE$-complete and $\BPOLY$ is $\NP$-complete by taking an arrangement of the matrices leading to a negative value as a certificate, 
and a polynomial-time verification procedure of this statement as a verifier.

\subsection{Stability of positive maps}
\label{app:STAB}

Let us now  review the reduction $\mathcal{R}: \MPO \to \STAB$ of \cite{Ey21}, 
which proves that $\STAB$ is $\RE$-hard. 
The same reduction also yields that $\BSTAB$ is $\NP$-hard. 

We map an instance 
$$
\langle B_1, \ldots, B_k \rangle \in \mathcal{M}_{D^2}(\mathbb{Q}) \cong \mathcal{M}_{D}(\mathbb{Q}) \otimes \mathcal{M}_D(\mathbb{Q})
$$ 
of $\MPO$ to a linear map
$$
\begin{array}{rrcl} \mathcal{P}:& \mathcal{M}_{D}(\mathbb{Q}) \otimes \mathcal{M}_D(\mathbb{Q}) & \to & \mathcal{M}_k(\mathbb{Q}) \\[0.2cm] 
& X & \mapsto & \displaystyle \sum_{i = 1}^{k} \ket{i} \bra{i} \tr(C_i X) \end{array} 
$$
where 
$$
(C_i)_{(\alpha_1, \alpha_2), (\beta_1, \beta_2)} \coloneqq (B_i)_{(\alpha_1, \beta_1), (\alpha_2, \beta_2)}
$$
with $\alpha_1, \alpha_2,\beta_1,\beta_2 \in \{1, \ldots, D\}$. Then, we have that
\begin{equation*}
\tr\Big(C_{i_1} \otimes \cdots \otimes C_{i_n} \chi_n\Big) = \tr(B_{i_1} \cdots B_{i_n})
\end{equation*}
where $\chi_n$ is defined in \eqref{eq:chin}. 
By construction, this implies that 
\begin{equation*}
\mathcal{P}^{\otimes n}(\chi_n) = \rho_n(B).
\end{equation*}

In summary, $\langle B_1, \ldots, B_k \rangle \in \MPO$ if and only if exists $n \in \mathbb{N}$ such that $\mathcal{P}^{\otimes n}(\chi_n) \ngeqslant 0$. Further the threshold parameters in both problems coincide for this reduction. 
It follows that $\BSTAB$ is $\NP$-hard.

\subsection{The tiling problem}
\label{app:Tiling}

We now review the reduction $\mathcal{R}: \HALT \to \TILE$ from \cite{Ro71}. A Turing machine, consisting of a tape alphabet $\Sigma$ with blank symbol \textvisiblespace \ $\in \Sigma$, a state set $Q$ with an initial state $q_{0}$ and final states $F \subseteq Q$, and a transition function
$$\delta: \Sigma \times (Q \setminus F) \to \Sigma \times Q \times \{L, R\}$$
is mapped to the following set of tiles:

\medskip\noindent
\renewcommand{\arraystretch}{2.5}
\setlength{\tabcolsep}{0.2em}
\begin{tabular}{r l c} (i) & Initial tile & \adjustbox{valign=m}{\begin{tikzpicture} 
\draw[very thick] (0,0) -- (0,1) -- (1,1) -- (1,0) -- cycle;
\draw[stealth-stealth] (0,0.5) -- (1,0.5);
\draw (0.5, 0.5) -- (0.5,1) node[fill=white] {\scriptsize $q_0$ \textvisiblespace};
\end{tikzpicture}} \\[0.25cm]
(ii) \label{ii} & Empty tape extension & \adjustbox{valign=m}{\begin{tikzpicture} 
\draw[very thick] (0,0) -- (0,1) -- (1,1) -- (1,0) -- cycle;
\draw[stealth-] (0,0.5) -- (0.5,0.5);
\draw[stealth-] (0.5,0.5) -- (1,0.5);
\draw (0.5, 0.5) -- (0.5,1) node[fill=white] {\scriptsize \textvisiblespace};
\end{tikzpicture}} \ and \ \adjustbox{valign=m}{\begin{tikzpicture} 
\draw[very thick] (0,0) -- (0,1) -- (1,1) -- (1,0) -- cycle;
\draw[-stealth] (0,0.5) -- (0.5,0.5);
\draw[-stealth] (0.5,0.5) -- (1,0.5);
\draw (0.5, 0.5) -- (0.5,1) node[fill=white] {\scriptsize \textvisiblespace};
\end{tikzpicture}}\\[0.25cm]
(iii) \label{iii} & Empty tile & \adjustbox{valign=m}{\begin{tikzpicture} 
\draw[very thick] (0,0) -- (0,1) -- (1,1) -- (1,0) -- cycle;
\end{tikzpicture}}\\[0.25cm]
(iv) \label{iv} & Transitions $(x,q) \mapsto (x', \hat{q},R)$ \hspace*{-0.8cm} & \adjustbox{valign=m}{\begin{tikzpicture} 
\draw[very thick] (0,0) -- (0,1) -- (1,1) -- (1,0) -- cycle;
\draw[stealth-] (0.5, 0.5) -- (0.5,0) node[fill=white] {\scriptsize $q x$};
\draw (0.5, 0.5) -- (0.5,1) node[fill=white] {\scriptsize $x'$};
\draw (0.5,0.5) -- (1,0.5) node[fill=white] {\scriptsize $\hat{q}$};
\end{tikzpicture}} \\[0.25cm]
(v) \label{v} & Transitions $(x,q) \mapsto (x', \hat{q},L)$ \hspace*{-0.8cm} & \adjustbox{valign=m}{\begin{tikzpicture} 
\draw[very thick] (0,0) -- (0,1) -- (1,1) -- (1,0) -- cycle;
\draw[stealth-] (0.5, 0.5) -- (0.5,0) node[fill=white] {\scriptsize $q x$};
\draw (0.5, 0.5) -- (0.5,1) node[fill=white] {\scriptsize $x'$};
\draw (0.5,0.5) -- (0,0.5) node[fill=white] {\scriptsize $\hat{q}$};
\end{tikzpicture}} \\[0.25cm]
(vi) \label{vi} & State merge & \adjustbox{valign=m}{\begin{tikzpicture} 
\draw[very thick] (0,0) -- (0,1) -- (1,1) -- (1,0) -- cycle;
\draw[stealth-] (0.5, 0.5) -- (0.5,0) node[fill=white] {\scriptsize $y$};
\draw (0.5, 0.5) -- (0.5,1) node[fill=white] {\scriptsize $\hat{q} y$};
\draw (0.5,0.5) -- (0,0.5) node[fill=white] {\scriptsize $\hat{q}$};
\end{tikzpicture}} \ and \ \adjustbox{valign=m}{\begin{tikzpicture} 
\draw[very thick] (0,0) -- (0,1) -- (1,1) -- (1,0) -- cycle;
\draw[stealth-] (0.5, 0.5) -- (0.5,0) node[fill=white] {\scriptsize $y$};
\draw (0.5, 0.5) -- (0.5,1) node[fill=white] {\scriptsize $\hat{q} y$};
\draw (0.5,0.5) -- (1,0.5) node[fill=white] {\scriptsize $\hat{q}$};
\end{tikzpicture}} \\[0.25cm]
(vii) \label{vii} & Copy tile for $x \in \Sigma$ & \adjustbox{valign=m}{\begin{tikzpicture} 
\draw[very thick] (0,0) -- (0,1) -- (1,1) -- (1,0) -- cycle;
\draw (0.5, 0.5) -- (0.5,1) node[fill=white] {\scriptsize $x$};
\draw[stealth-] (0.5, 0.5) -- (0.5,0) node[fill=white] {\scriptsize $x$};
\end{tikzpicture}}
\end{tabular}

\medskip
Note that \hyperref[vi]{(vi)}, State merge, is defined for every $y\in \Sigma$ and  $\hat q\in Q$, 
whereas  \hyperref[iv]{(iv)}  and  \hyperref[v]{(v)}, Transitions, are defined for every such tuple in $\delta$. 
 
This set of tiles captures the computation of a Turing machine on the empty tape when placing the initial tile to the origin (see \Cref{fig:Tiling}). 
The initial tile can only be extended to the left and to the right with \hyperref[ii]{(ii)} Empty tape extensions. 
We can also trivially tile the whole lower half of the plane by applying the empty tile. 
The generated string
$$ 
\ldots \quad \text{\textvisiblespace \quad \textvisiblespace \quad \textvisiblespace } \quad q_0 \; \text{\textvisiblespace \quad \textvisiblespace \quad \textvisiblespace \quad \textvisiblespace} \quad \ldots 
$$
at the top of the first line represents the instantaneous description of the Turing machine at time $0$, namely an empty tape with the head at position $0$ and state $q_0$. 
Simulating one step of the Turing machine corresponds to filling up the line above of the current one. 
Specifically, on the top of the initial tile, we need to place a Transition tile (\hyperref[iv]{(iv)}  or  \hyperref[v]{(v)})
$(q_0, \text{\textvisiblespace}) \mapsto (\hat{q}, x, L/R)$. 
Then we need to place a  \hyperref[vi]{(vi)}  State merge tile on the left/right of the transition tile. 
This reflects the movement of the head to the left or right. 
The rest of the line is filled with  \hyperref[vii]{(vii)} Copy tiles.

Again, the string at the top of the second line represents the initial description after one computation step. The same procedure applies to every computation step. As soon as we apply a transition tile $(q,x) \mapsto (q_f, y, L/R)$ for some final state $q_f \in F$, there is no tile to continue the tiling procedure. In other words, every tiling procedure terminates in line $n$ if and only if $T$ halts on the empty tape.

The same reduction applies to non-deterministic Turing machines. 
In this situation, every tiling procedure terminates in $n$ lines if and only if 
the Turing machine halts on the empty tape along every computation path in at most $n$ steps. 
In other words, a Turing machine $T$ halts on every path in at most $n$ steps 
if and only if $\mathbb{Z}_{n+1} \times \mathbb{Z}_{n+1}$ cannot be tiled. 
This proves that $\mathcal{R}: \NHALTAll \to \TILE$ is a reduction. 
It follows that $\TILE$ is $\RE$-hard. 

Moreover, $\mathcal{R}$ is a polynomial-time map. Since the map between the threshold parameters of $\NHALTAll$ and $\TILE$ is given by $n \mapsto n+1$, 
$$
\langle x,1^n\rangle \mapsto \langle \mathcal{R}(x), 1^{n+1}\rangle
$$ 
is a reduction from $\BNHALTAll$ to $\BTILE$. 
This implies that $\BTILE$ is $\coNP$-hard. 

\end{document}